\documentclass[11pt,a4paper]{article}

\usepackage{algorithm}
\usepackage[noend]{algpseudocode}

\usepackage[american]{babel}

\usepackage{dsfont}
\usepackage{amsthm}
\usepackage{amsmath}
\usepackage{amssymb} %,bm,bbm}
\usepackage{xspace}
\usepackage{color}
\usepackage{enumerate}
\usepackage[symbol]{footmisc}%footnotes are indexed by symbols, not numbers

\usepackage{thmtools, thm-restate}
\declaretheorem[numberwithin=section]{theorem}

\declaretheorem[sibling=theorem]{lemma}

\declaretheorem[sibling=theorem]{definition}

\usepackage{multirow}

\usepackage[latin1]{inputenc}

\newcommand{\R}{\mathbb{R}}
\newcommand{\Z}{\mathbb{Z}}
\newcommand{\N}{\mathbb{N}}

\newcommand{\Vol}{\textsc{vol}}

\newcommand{\Ex}{\mathbb{E}}

\newcommand{\barw}{\ensuremath{\bar w}}
\newcommand{\barV}{\ensuremath{\bar V}}

\newcommand{\ceilpowtwo}[1]{\lceil #1 \rceil_{2^d}}
\newcommand{\parent}{\textup{par}}

\newcommand{\x}[1]{\ensuremath{\mathsf{x}_{#1}}}
\newcommand{\w}[1]{\ensuremath{\mathsf{w}_{#1}}}
\newcommand{\W}{\ensuremath{\mathsf{W}}}
\newcommand{\wmin}{\ensuremath{\mathsf{w}_{\min}}}
\newcommand{\wmax}{\ensuremath{\mathsf{w}_{\max}}}

\newcommand{\cc}{\textsc{cc}}

\newcommand{\Geo}{\textup{Geo}}

\newcommand{\Space}{\ensuremath{\mathds{T}^d}}

\newcommand{\Select}{\textsc{Select}}
\newcommand{\Rank}{\textsc{Rank}}

\renewcommand{\epsilon}{\ensuremath{\varepsilon}}
\newcommand{\eps}{\ensuremath{\varepsilon}}

\newcommand{\temporary}[1]{}

\makeatletter
\newcommand{\pushright}[1]{\ifmeasuring@#1\else\omit\hfill$\displaystyle#1$\fi\ignorespaces}
\newcommand{\pushleft}[1]{\ifmeasuring@#1\else\omit$\displaystyle#1$\hfill\fi\ignorespaces}
\makeatother

\begin{document}

\title{Sampling Geometric Inhomogeneous Random Graphs in\\  Linear Time}
\author{Karl Bringmann\thanks{Max-Planck-Institute for Informatics, Saarbr\"ucken, Germany, \texttt{kbringma@mpi-inf.mpg.de}} \and Ralph Keusch\thanks{Institute of Theoretical Computer Science, ETH Zurich, Switzerland, \texttt{rkeusch@inf.ethz.ch}} \and Johannes Lengler\thanks{Institute of Theoretical Computer Science, ETH Zurich, Switzerland, \texttt{lenglerj@inf.ethz.ch}}}
\date{}
\maketitle

%\vspace{-0.5cm}
%\begin{center}
%%\large
%  Institute of Theoretical Computer Science \\
%  ETH Zurich, Switzerland
%%\vspace{8mm}
%\end{center}

%\medskip
%
%\begin{abstract}
%
%\begin{center}
%\LARGE Geometric Inhomogeneous Random Graphs
%\vspace{8mm}
%
%\Large{
%\begin{tabular}{ccccc}
%Karl Bringmann%\footnote{Supported by an ETH Zurich Postdoctoral Fellowship.}
%& \quad & Ralph Keusch & \quad & Johannes Lengler \\
%{\small{karl.bringmann@inf.ethz.ch}} & \quad & {\small{rkeusch@inf.ethz.ch}} & \quad & {\small{johannes.lengler@inf.ethz.ch}}
%\end{tabular}
%}
%\vspace{5mm}
%
%\large
%  Institute of Theoretical Computer Science \\
%  ETH Zurich, 8092 Zurich, Switzerland
%\vspace{8mm}
%\end{center}

\medskip

\begin{abstract}
   Real-world networks, like social networks or the internet infrastructure, have structural properties such as large clustering coefficients that can best be described in terms of an underlying geometry. This is why the focus of the literature on theoretical models for real-world networks shifted from classic models without geometry, such as Chung-Lu random graphs, to modern geometry-based models, such as hyperbolic random graphs. 
  
  With this paper we contribute to the theoretical analysis of these modern, more realistic random graph models. Instead of studying directly hyperbolic random graphs, we use a generalization that we call \emph{geometric inhomogeneous random graphs} (GIRGs). 
  Since we ignore constant factors in the edge probabilities, GIRGs are technically simpler (specifically, we avoid hyperbolic cosines), while preserving the qualitative behaviour of hyperbolic random graphs, and we suggest to replace hyperbolic random graphs by this new model in future theoretical studies.
  
  We prove the following fundamental structural and algorithmic results on GIRGs. (1) As our main contribution we provide a sampling algorithm that generates a random graph from our model in expected \emph{linear} time, improving the best-known sampling algorithm for hyperbolic random graphs by a substantial factor $O(\sqrt{n})$. (2) We establish that GIRGs have clustering coefficients in $\Omega(1)$, (3) we prove that GIRGs have small separators, i.e., it suffices to delete a sublinear number of edges to break the giant component into two large pieces, and (4) we show how to compress GIRGs using an expected linear number of bits.
\end{abstract}

\thispagestyle{empty}
\newpage
\setcounter{page}{1}

\section{Introduction}
Real-world networks, like social networks or the internet infrastructure, have structural properties that can best be described using \emph{geometry}. 
For instance, in social networks two people are more likely to know each other if they live in the same region and share hobbies, both of which can be encoded as spatial information.  
This geometric structure may be responsible for some of the key properties of real-world networks, e.g., an underlying geometry naturally induces a large number of triangles, or large \emph{clustering coefficient}: Two of one's friends are likely to live in one's region and have similar hobbies, so they are themselves similar and thus likely to know each other. 

Almost always, large real-world networks are
scale-free, i.e., their degree distribution follows a power law. Such networks have been studied in detail since the 60s. One of the key findings is the small-world phenomenon, which is the observation that two nodes in a network typically have very small graph-theoretic distance. Therefore, classic mathematical models for real-world networks reproduce these two key findings. But since they have no underlying geometry their clustering coefficient is as small as~$n^{-\Omega(1)}$; this holds in particular for preferential attachment graphs~\cite{barabasi1999emergence} and Chung-Lu random graphs \cite{chung2002avg,Chung02,chung2004avg} (and their variants 
\cite{BollobasSR07,NorrosReittu06}). 
In order to close this gap between the empirically observed clustering coefficient and theoretical models, much of the recent work on models for real-world networks focussed on scale-free random graph models that are equipped with an underlying geometry, such as hyperbolic random graphs \cite{BogunaPK10,PapadopoulosKBV10}, spatial preferred attachment \cite{spatialpreferred}, and many others \cite{BollobasSR07,bonato2010geometric,bradonjic2008structure,
jacob2013spatial}. The basic properties -- scale-freeness, small-world, and large clustering coefficient -- have been rigorously established for most of these models.
Beyond the basics, experiments suggest that these models have some very desirable properties. 

In particular, hyperbolic random graphs are a promising model, as Bogu\~n\'a et al.~\cite{BogunaPK10} computed a (heuristic) maximum likelihood fit of the internet graph into the hyperbolic random graph model and demonstrated its quality by showing that greedy routing in the underlying geometry of the fit finds near-optimal shortest paths. 
Further properties that have been studied on hyperbolic random graphs, mostly agreeing with empirical findings on real-world networks, are scale-freeness and clustering coefficient~\cite{gugelmann2012random,DBLP:conf/waw/CandelleroF14}, existence of a giant component~\cite{bode2013giant}, diameter~\cite{kiwi2015bound,FriedrichKrohmer15}, average distance~\cite{abdullah2015typical}, separators and treewidth~\cite{blasius2016hyperbolic}, bootstrap percolation~\cite{DBLP:conf/waw/CandelleroF14}, and clique number~\cite{friedrich2015cliques}. Algorithmic aspects include sampling algorithms~\cite{LoozSMP15} and compression schemes~\cite{uelithesis}. 

Our goal is to improve algorithmic and structural results on the promising model of hyperbolic random graphs. However, it turns out to be beneficial to work with a more general model, which we introduce with this paper: 
In a \emph{geometric inhomogeneous random graph} (GIRG), every vertex $v$ comes with a weight~$\w v$ (which we assume to follow a power law in this paper) and picks a uniformly random position $\x v$ in the $d$-dimensional torus\footnote{We choose a toroidal ground space for the technical simplicity that comes with its symmetry and in order to obtain hyperbolic random graphs as a special case. The results of this paper stay true if $\Space$ is replaced, say, by the $d$-dimensional unitcube $[0,1]^d$.} $\Space$. Two vertices $u,v$ then form an edge independently with probability $p_{uv}$, which is proportional to $\w u \w v$ and inversely proportional to some power of their distance $\|\x u - \x v\|$, see Section~\ref{sec:modelresults} for details\footnote{A major difference between hyperbolic random graphs and our generalisation is that we ignore constant factors in the edge probabilities $p_{uv}$. This allows to greatly simplify the edge probability expressions, thus reducing the technical overhead.}. The model can be interpreted as a geometric variant of the classic Chung-Lu random graphs.  Recently, with scale-free percolation a related model has been introduced~\cite{deijfen2013scale} where the vertex set is given by the grid $\Z^d$. This model is similar with respect to component structure, clustering, and small-world properties~\cite{deprez2015percolation, heydenreich2016structures}, but none of the algorithmic aspects studied in the present paper (sampling, compression, also separators) has been regarded thereon.

%A major difference between GIRGs and hyperbolic random graphs, which we prove to be a special case of GIRGs, is that we ignore constant factors in the edge probabilities $p_{uv}$. This allows to greatly simplify the edge probability expressions, thus reducing the technical overhead. In particular, proving the results of this paper directly for hyperbolic random graphs would have been much more tedious.  This is why we suggest GIRGs as a replacement for hyperbolic random graphs in future theoretical studies.

The basic connectivity properties of GIRGs follow from more general considerations in~\cite{bringmann2015generalGIRG}, where a model of generic augmented Chung-Lu graphs is studied. In particular, with high probability\footnote{We say that an event holds \emph{with high probability} (whp) if it holds with probability $1 - n^{-\omega(1)}$.} GIRGs have a giant component and polylogarithmic diameter, and almost surely doubly-logarithmic average distance. However, general studies such as~\cite{bringmann2015generalGIRG} are limited to properties that do not depend on the specific underlying geometry. Very recently, GIRGs turned out to be accesible for studying processes such as bootstrap percolation \cite{koch2016bootstrap} and greedy routing \cite{bringmann2017routing}.

\paragraph*{Our contribution:} Models of complex networks play an important role in algorithm development and network analysis as typically real data is scarce. Regarding applications and simulations, it is crucial that a model can be generated fast in order to produce sufficiently large samples. As our main result, we present a sampling algorithm that generates a random graph from our model in expected linear time. This improves the trivial sampling algorithm by a factor $O(n)$ and the best-known algorithm for hyperbolic random graphs by a substantial factor $O(\sqrt{n})$~\cite{LoozSMP15}.

We also prove that the underlying geometry indeed causes GIRGs to have a clustering coefficient in $\Omega(1)$. Moreover, we show that GIRGs have small separators of expected size $n^{1-\Omega(1)}$; this is in agreement with empirical findings on real-world networks~\cite{blandford2003compact}. We then use the small separators to prove that GIRGs can be efficiently compressed (i.e., they have low entropy); specifically, we show how to store a GIRG using $O(n)$ bits in expectation. %Our representation is efficient in that it allows query the $i$-th neighbor of any vertex in time~$O(1)$.
Finally, we show that hyperbolic random graphs are indeed a special case of GIRGs, so that all aforementioned results also hold for hyperbolic random graphs.

\paragraph*{Organization of the paper:} We present the details of the model and our results in Section~\ref{sec:modelresults}. In Section~\ref{sec:preliminaries} we introduce notation and a geometric ordering of the vertices, and we present some basic properties of the GIRG model. Afterwards, we prove our main result on sampling algorithms in Section~\ref{sec:sampling}. We analyze the clustering coefficient in Section~\ref{sec:clustering}, and determine the separator size and the entropy in Section~\ref{sec:entropy}. Finally in Section~\ref{sec:hyperbolic} we establish that hyperbolic random graphs are a special case of GIRGs, and in Section~\ref{sec:conclusion} we make some concluding remarks.

\section{Model and Results}
\label{sec:modelresults}

\subsection{Definition of the Model}\label{subsec:model}
We prove algorithmic and structural results in a new random graph model which we call \emph{geometric inhomogeneous random graphs}. In this model,  each vertex $v$ comes with a weight $\w v$ and with a random position $\x v$ in a geometric space, and the set of edges $E$ is also random.
We start by defining the by-now classical Chung-Lu model and then describe the
changes that yield our variant with underlying geometry.

\paragraph*{Chung-Lu random graph:}
For $n \in \N$ let $\w{} = (\w 1, \ldots, \w n)$ be a sequence of positive weights.
We call $\W := \sum_{v=1}^n \w v$ the \emph{total weight}. The Chung-Lu random graph $G(n,\w{})$ has
vertex set $V = [n] = \{1,\ldots,n\}$, and two vertices $u \ne v$ are connected by
an edge independently with probability $p_{uv} = \Theta\big(\min\big\{1,
\frac{\w u \w v}{\W}\big\}\big)$ \cite{chung2002avg,Chung02}. Note that the term
$\min\{1,.\}$ is necessary, as the product $\w u \w v$ may be larger than $\W$.
Classically, the $\Theta$ simply hides a factor 1, but by introducing the
$\Theta$ the model also captures similar random graphs, like the
Norros-Reittu model \cite{NorrosReittu06}, while important properties stay
asymptotically invariant.

\paragraph*{Geometric inhomogeneous random graph (GIRG):}

Note that we obtain a circle by identifying the endpoints of the interval $[0,1]$. Then the distance of $x,y \in [0,1]$ along the circle is $|x-y|_C := \min\{|x-y|,1-|x-y|\}$. 
We fix a dimension $d \ge 1$ and use as our \emph{ground space} the $d$-dimensional torus $\Space=\R^d / \Z^d$, which can be described as the $d$-dimensional cube $[0,1]^d$ where opposite boundaries are identified. As distance function we use the $\infty$-norm on $\Space$, i.e., for $x,y \in \Space$ we define $\|x-y\| := \max_{1 \le i \le d} | x_i - y_i|_C$. 

%Any norm $\|.\|$ on $\R^d$ induces a translation invariant metric on the torus $\Space$ by setting the distance of $x,y \in \Space$ to $\|x - y + u\|$, where $u \in \Z^d$ is the unique point with $x-y+u \in [-\frac 12,\frac 12]^d$. As we will always work on the torus, with slight abuse of notation we call the induced metric again $\|x-y\|$. We fix a dimension $d \ge 1$ and a norm $\|.\|$, and let the resulting metric space $(\Space,\|.\|)$ be our ground space. 

As for Chung-Lu graphs, we consider the vertex set $V = [n]$ and a weight sequence $\w{}$ (in this paper we require the weights to follow a power law with exponent $\beta > 2$, see next paragraph). Additionally, for any vertex $v$ we draw a point $\x v \in \Space$ uniformly and independently at random.
Again we connect vertices $u \ne v$ independently with probability $p_{uv} = p_{uv}(r)$, which now depends not only on the weights $\w u,\w v$ but also on the positions~$\x u,\x v$, more precisely, on the distance $r=\|\x u - \x v\|$. We require for some constant $\alpha > 1$ the following edge probability condition:
\begin{equation}\label{eq:puv} p_{uv} = \Theta\bigg( \min\bigg\{ \frac{1}{\|\x u - \x v\|^{\alpha d}} \cdot \Big( \frac{\w u \w v}{\W}\Big)^{\alpha} ,1 \bigg\} \bigg). \tag{EP1}
\end{equation}
We also allow $\alpha = \infty$ and in this case require that %there are constants $0 < c_1 \le c_2$ such that
\begin{equation}\label{eq:puv2}
 p_{uv} = \begin{cases} \Theta(1), & \text{if } \|\x u - \x v\| \le O\big(\big(\tfrac{\w u \w v}\W\big)^{1/d}\big), \\ 0, & \text{if } \|\x u - \x v\| \ge \Omega\big(\big(\tfrac{\w u \w v}\W\big)^{1/d}\big), \end{cases}  \tag{EP2}
 \end{equation}
where the constants hidden by $O$ and $\Omega$ do not have to match, i.e., there can be an interval $[c_1 (\tfrac{\w u \w v}\W)^{1/d}, c_2 (\tfrac{\w u \w v}\W)^{1/d}]$ for $\|\x u - \x v\|$ where the behaviour of $p_{uv}$ is arbitrary.
%note that this is the limit for $\alpha \to \infty$.
This finishes the definition of GIRGs. The free parameters of the model are $\alpha \in (1,\infty]$, $d \in \N$, the concrete weights~$\w{}$ with power-law exponent $\beta > 2$ and average weight $\W/n$, the concrete function $f_{uv}(\x u, \x v)$ replacing the $\Theta$ in $p_{uv}$, and for $\alpha = \infty$ the constants hidden by $O,\Omega$ in the requirement for $p_{uv}$. We will typically hide the constants $\alpha,d,\beta,\W/n$ by $O$-notation.

\paragraph*{Power-law weights:}

As is often done for Chung-Lu graphs, we assume throughout this paper that the weights follow a \emph{power law}: the fraction of vertices with weight at least $w$ is proportional to $w^{1-\beta}$ for some $\beta>2$ (the \emph{power-law exponent} of $\w{}$). More precisely, we assume that for some $\barw = \barw(n)$ with $n^{\omega(1/\log\log n)}\leq\barw \leq n^{(1-\Omega(1))/(\beta-1)}$, the sequence $\w{}$ satisfies the following conditions:
\begin{enumerate}[(PL1)]
\item the minimum weight is constant, i.e., $\wmin := \min\{\w{v} \mid 1 \le v \le n\} = \Omega(1)$, % and $\wmax := \max\{w_i \mid i\in [n]\} = n^{1-\Omega(1)}$.
\item for all $\eta >0$ there exist constants $c_1,c_2>0$ such that
\[
c_1\frac{n}{w^{\beta-1+\eta}} \leq \#\{1 \le v \le n \mid \w{v} \geq w\} \leq c_2\frac{n}{w^{\beta-1-\eta}},
\]
where the first inequality holds for all $\wmin \leq w \leq \barw$ and the second for all $w \geq \wmin$. 
\end{enumerate}

We remark that these are standard assumptions for power-law graphs with average degree
$\Theta(1)$.
In particular, (PL2) implies that the average weight $\W / n$ is $\Theta(1)$.
An example is the widely used weight function $\w v := \delta\cdot
(n/v)^{1/(\beta-1)}$ with parameter $\delta= \Theta(1)$. 
%For $\alpha = \infty$ we have to strengthen the assumption on $\barw$ to $\barw(n) = \omega(n^{1/2})$. This is necessary to obtain a meaningful bound on diameter and average distance, as we discuss in Appendix~\ref{:diameter}.

\paragraph*{Discussion of the model:}
The choice of the ground space $\Space$ is in the spirit of the classic random geometric graphs~\cite{penrose2003}. We prefer the torus to the hyper-cube for technical simplicity, as it yields symmetry. However, one could replace $\Space$ by $[0,1]^d$ or other manifolds like the $d$-dimensional sphere; our results will still hold verbatim.
Moreover, since in fixed dimension all $L_p$-norms on $\Space$ are equivalent and since the edge probabilities $p_{uv}$ have a constant factor slack, our choice of the $L_\infty$-norm is without loss of generality (among all norms).

The model is already motivated since it generalizes the celebrated hyperbolic random graphs (see Section~\ref{sec:hyperbolic}).
Let us nevertheless discuss why our choice of edge probabilities is \emph{natural}: The term $\min\{.,1\}$ is necessary, as in the Chung-Lu model, because $p_{uv}$ is a probability. To obtain a geometric model, where adjacent vertices are likely to have small distance, $p_{uv}$ should decrease with increasing distance $\|\x u - \x v\|$, and an inverse polynomial relation seems reasonable. The constraint $\alpha > 1$ is necessary to cancel the growth of the volume of the ball of radius $r$ proportional to $r^d$, so that we expect most neighbors of a vertex to lie close to it. Finally, the factor $\big( \frac{\w u \w v}{\W}\big)^{\alpha}$ ensures that the marginal probability of vertices $u,v$ with weights $\w u,\w v$ forming an edge is 
$\Pr[u \sim v] = \Theta\left(\min\left\{\frac{\w u \w v}{\W},1 \right\} \right)$,
as in the Chung-Lu model, and this probability does not change by more than a constant factor if we fix either $\x{u}$ or $\x{v}$. This is why we see our model as a geometric variant of Chung-Lu random graphs.
For a fixed vertex $u \in V$ we can sum up $\Pr[u \sim v \mid \x u]$ over all vertices $v \in V \setminus \{u\}$, and it follows 
$\Ex[\deg(u)]=\Theta(\w{u})$.
The main reason why GIRGs are also \emph{technically easy} is that for any vertex $u$ with fixed position $\x u$ the incident edges $\{u,v\}$ are independent. Details of these basic properties can be found in Section~\ref{subsec:basicproperties}.

\paragraph*{Sampling the weights:} In the definition we assume that the weight sequence $\w{}$ is fixed. However, if we sample the weights independently according to an appropriate power-law distribution with minimum weight $\wmin$ and density
$f(w) \sim w^{-\beta}$,
then for a given $\eta>0$ the sampled weight sequence will follow a power law and fulfils (PL1) and (PL2) with probability $1-n^{-\Omega(1)}$. Hence, a model with sampled weights is almost surely included in our model. For the precise statement, see Lemma~\ref{lem:sampleweights}.

\subsection{Structural Properties of GIRGs}\label{subsec:strucprop}

As discussed in the introduction, reasonable random graph models for real-world networks should reproduce a power-law degree distribution and small graph-theoretical distances between nodes. Before giving a detailed list of our results in Section~\ref{subsec:results}, we first want to ensure that GIRGs have these desired structural properties. Indeed, they follow from a more general class of generic augmented Chung-Lu random graphs that have been studied in \cite{bringmann2015generalGIRG}. This framework has weaker assumptions on the underlying geometry than GIRGs. A short comparison reveals that GIRGs are a special case of this general class of random graph models. In the following we list the results of \cite{bringmann2015generalGIRG} transferred to GIRGs. As we are using power-law weights and $\Ex[\deg(v)]=\Theta(\w{v})$ holds for all $v \in V$, it is not surprising that the degree sequence follows a power-law.

\begin{theorem}[Theorem~2.1 in \cite{bringmann2015generalGIRG}] \label{thm:degseq1}
  Whp the degree sequence of a GIRG follows a power law with exponent~$\beta$ and average degree $\Theta(1)$. 
\end{theorem}

The next result determines basic connectivity properties. Note that for $\beta > 3$, they are not well-behaved, in particular since in this case even threshold hyperbolic random graphs do not possess a giant component of linear size~\cite{bode2014geometrisation}. Therefore, $2<\beta<3$ is assumed, which is the typical regime of empirical data. For the following theorem, we require the additional assumption $\barw = \omega(n^{1/2})$ in the limit case $\alpha=\infty$.

\begin{theorem}[Theorems~2.2 and 2.3 in \cite{bringmann2015generalGIRG}]\label{thm:diameter}
 Let $2<\beta<3$. Then whp the largest component of a GIRG has linear size and diameter $\log^{O(1)} n$, while all other components have size $\log^{O(1)} n$. Moreover, the average distance of vertices in the largest component is $(2 \pm o(1))\frac{\log \log n}{|\log(\beta-2)|}$ in expectation and with probability $1-o(1)$. 
\end{theorem}

We remark that most results of this paper crucially depend on an underlying geometry, and thus do \emph{not} hold in the general model from~\cite{bringmann2015generalGIRG}.

\subsection{Results}
\label{subsec:results}

\paragraph*{Sampling:}
Sampling algorithms that generate a random graph from a fixed distribution are known for Chung-Lu random graphs and others, running in expected linear time~\cite{batagelj_efficient_2005,miller_efficient_2011}.
As our main result, we present such an algorithm for GIRGs. This greatly improves the trivial $O(n^2)$ sampling algorithm (throwing a biased coin for each possible edge), as well as the best previous algorithm for threshold hyperbolic random graphs with expected time $O(n^{3/2})$~\cite{LoozSMP15}. It allows to run experiments on much larger graphs than the ones with $\approx 10^4$ vertices in~\cite{BogunaPK10}. 
In addition to our model assumptions, here we assume that the $\Theta$ in our requirement on $p_{uv}$ is sufficiently explicit, i.e., we can compute $p_{uv}$ exactly and we know a constant $c>0$ such that replacing $\Theta$ by $c$ yields an upper bound on $p_{uv}$, see Section~\ref{sec:sampling} for details.

\begin{theorem}[Section~\ref{sec:sampling}]\label{thm:sampling}
  Sampling a GIRG can be done in expected time $O(n)$. 
\end{theorem}

%\begin{table*}
%
%\label{tab:compare}
%\renewcommand{\arraystretch}{2}
%\begin{center}
%\begin{tabular}{l c c c}
%\hline
%\bfseries Random Graph Model & \bfseries Chung-Lu & \bfseries Hyperbolic (old) & \bfseries GIRG\\
%\hline\hline
%\bfseries Clustering coefficient & $n^{-\Omega(1)}$ & $\Theta(1)$ & $\Theta(1)$\\
%\hline
%\bfseries Expected sampling time & $O(n)$ & $O(n^{3/2})$ & $O(n)$\\
%\hline
%\bfseries Size of seperators & $\Theta(n)$ & no results & $n^{1-\Omega(1)}$\\
%\hline
%\bfseries Entropy & $\Theta(n \log n)$ & $\Theta(n)$ & $\Theta(n)$\\
%\hline
%\end{tabular}
%\end{center}
%
%
%\vspace{0.2cm}
%Table 1: Comparison of our results on GIRGs to other scale-free random graph models. Note that our results transfer to hyperbolic random graphs and therefore improve previous results on sampling and stability of this model. 
%\end{table*}

\paragraph*{Clustering:}
In social networks, two friends of the same person are likely to also be friends with each other. This property of having many triangles is captured by the clustering coefficient, defined as the probability when choosing a random vertex $v$ and two random neighbors $v_1 \ne v_2$ of $v$ that $v_1$ and $v_2$ are adjacent (if $v$ does not have two neighbors then its contribution to the clustering coefficient is 0). While Chung-Lu random graphs have a very small clustering coefficient of $n^{-\Omega(1)}$, it is easy to show that the clustering coefficient of GIRGs is $\Theta(1)$. This is consistent with empirical data of real-world networks~\cite{dorogovtsev2002evolution} and the constant clustering coefficient of hyperbolic random graphs determined in~\cite{DBLP:conf/waw/CandelleroF14,gugelmann2012random,uelithesis}.

\begin{theorem}[Section~\ref{sec:clustering}] \label{thm:clustering}
  Whp the clustering coefficient of a GIRG is $\Theta(1)$. 
\end{theorem}
%\begin{proof}[Proof outline]
%  We show that the clustering coefficient is dominated by the contribution of constant-weight vertices $v$. Let $v \in V$ be a vertex of weight $\w{v}=\Theta(1)$. Then, with at least constant probability, (i) $\deg(v)\ge 2$, and (ii) all neighbors of $v$ are located in a ball of radius $cn^{-1/d}$ around $\x{v}$, for a sufficiently small constant $c>0$. If the neighborhood of $v$ has this property, then two random neighbors $v_1$, $v_2$ of $v$ are connected with constant probability. Therefore, the expected contribution of $v$ to the clustering coefficient is $\Omega(\frac{1}{n})$. As the number of such vertices $v$ is $\Theta(n)$, it follows that the expected clustering coefficient is $\Theta(1)$. 
%  Proving the whp-statement requires additional arguments and the application of Azuma-type concentration inequalities with bad events. We defer the full proof to the appendix (see Section~\ref{sec:clustering}).
%\end{proof}

\paragraph*{Stability:}
For real-world networks, a key property to analyze is their stability under attacks. It has been empirically observed that many real-world networks have small separators of size $n^{c}$, $c<1$~\cite{blandford2003compact}. In contrast, Chung-Lu random graphs are unrealistically stable, since any deletion of $o(n)$ nodes or edges reduces the size of the giant component by at most $o(n)$~\cite{BollobasSR07}.
We show that GIRGs agree with the empirical results much better. Specifically, if we cut the ground space $\Space$ into two halves along one of the axes then we roughly split the giant component into two halves, but the number of edges passing this cut is quite small, namely $n^{1-\Omega(1)}$.
Thus, GIRGs are prone to (quite strong) adversarial attacks, just as many real-world networks. Furthermore, their small separators are useful for many algorithms, e.g., the compression scheme of the next paragraph. 

%Thus, these graphs are prone to adversarial attacks. In particular, GIRGs and thus also hyperbolic random graphs may not be the right random graph model for some applications. However, small separators also have positive effects, see the paragraph on entropy below.
%We remark that stability was not studied for (threshold) hyperbolic random graphs before.
%\jl{I would phrase this section positively, with main message ``GIRGs have small separators, which agrees with empirical data. This is useful for many algorithms.... The exact size of the separators is a crucial property of networks since too small separators make them prone to attacks.''}

\begin{theorem}[Section~\ref{sec:entropy}]\label{thm:stability}
Let $2<\beta<3$. Then almost surely it suffices to delete $$O\left(n^{\max\{2-\alpha, 3-\beta, 1-1/d \}+o(1)}\right)$$ edges of a GIRG to split its giant component into two parts of linear size each.
\end{theorem}

Since we assume $\alpha > 1$, $\beta > 2$, and $d = \Theta(1)$, the number of deleted edges is indeed $n^{1-\Omega(1)}$. Recently, Bl\"{a}sius et al.\ \cite{blasius2016hyperbolic} proved a better bound of $O(n^{(3-\beta)/2})$ for threshold hyperbolic random graphs which correspond to GIRGs with parameters $d=1$ and $\alpha=\infty$.

\paragraph*{Entropy:}

The internet graph has empirically been shown to be well compressible, using only 2-3 bits per edge \cite{blandford2003compact,BoldiV03}. This is not the case for the Chung-Lu model, as its entropy is $\Theta(n \log n)$ \cite{ChierichettiKLPR09}. We show that GIRGs have linear entropy, as is known for threshold hyperbolic random graphs~\cite{uelithesis}.

%\kb{which may explain the compressibility of real-world networks. In fact, the basic compression algorithm of Boldi and Vigna~\cite{BoldiV03} is to find a suitable ordering of the vertices and then store an edge between the vertices at positions $i,j$ by storing the offset $i-j$ at vertex $i$, thus using $O(\log(|i-j|))$ bits. Using space-filling curves, for any GIRG we construct an ordering of the vertices for which this compression scheme uses $O(n)$ bits (in expectation and w.h.p.).}

\begin{theorem}[Section~\ref{sec:entropy}]\label{thm:entropy}
  We can store a GIRG using $O(n)$ bits in expectation. The resulting data structure allows to query the degree of any vertex and its $i$-th neighbor in time $O(1)$. The compression algorithm runs in time $O(n)$.
\end{theorem}

\paragraph*{Hyperbolic random graphs:}
We establish that hyperbolic random graphs are an example of one-dimensional GIRGs, and that the often studied special case of threshold hyperbolic graphs is obtained by our limit case $\alpha = \infty$. Specifically, we obtain hyperbolic random graphs from GIRGs by setting the dimension $d=1$, the weights to a specific power law, and the $\Theta$ in the edge probability $p_{uv}$ to a specific, complicated function.
\begin{theorem}[Section~\ref{sec:hyperbolic}]\label{thm:hyperbolic}
For every choice of parameters in the hyperbolic random graph model, there is a choice of parameters in the GIRG model such that the two resulting distributions of graphs coincide.
\end{theorem}

In particular, all our results on GIRGs hold for hyperbolic random graphs, too. 
%, which is why in the following we compare our results on GIRGs with the literature on (threshold) hyperbolic random graphs. 
Moreover, as our proofs are much less technical than typical proofs for hyperbolic random graphs, we suggest to switch from hyperbolic random graphs to GIRGs in future studies.

\section{Preliminaries}\label{sec:preliminaries}

%We start with some preliminaries. In Section~\ref{subsec:notation} the reader can find a condensed list of our notation, and in Section~\ref{sec:tools} we give some standard theorems which we use in the proofs.
\subsection{Notation}\label{subsec:notation}
%We use the Landau notations $O(.)$, $o(.)$, $\Omega(.)$, $\omega(.)$, and $\Theta(.)$ in the usual sense for $n \to \infty$. 
%We use $o_w(f)$ for a function that is $o(f)$ with respect to $w$, i.e., for $w \to \infty$, and similarly for other indices. 

For $w\in \R_{\geq 0}$, we use the notation
%\begin{align*}
$  V_{\geq w} := \{v\in V\; |\; \w{v}\geq w\}$ and $V_{\leq w} := \{v\in V\; |\; \w{v}\leq w\}$,
%\end{align*}
as well as
$
  \W_{\geq w}:=\sum_{v\in V_{\geq w}} \w{v}$ and $\W_{\leq w}:=\sum_{v\in V_{\geq w}}\w{v}.
$
For $u,v\in V$ we write $u\sim v$ if $u$ and $v$ are adjacent, and for
$A,B\subseteq V$ we write $A\sim v$ if there exists $u\in A$ such that $u\sim v$, and we write $A \sim B$ if there exists $v \in B$ such that $A\sim v$. For a vertex
$v\in V$, we denote its neighborhood by $\Gamma(v)$, i.e.\ $\Gamma(v):=\{u\in
V\mid u\sim v\}$. We say that an event holds \emph{with high probability} (whp) if it holds with probability $1-n^{-\omega(1)}$.

\subsection{Cells}\label{subsec:cells}

We introduce a geometric ordering of the vertices, which we will use both for the sampling and for the compression algorithms. Consider the ground space $\Space$, split it into $2^d$ equal cubes, and repeat this process with each created cube; we call the resulting cubes \emph{cells}.
Cells are cubes of the form $C = [x_1 2^{-\ell}, (x_1+1) 2^{-\ell}) \times \ldots \times [x_d 2^{-\ell}, (x_d+1)2^{-\ell})$ with $\ell \ge 0$ and $0 \le x_i < 2^\ell$. We represent cell $C$ by the tuple $(\ell,x_1,\ldots,x_d)$. 
The volume of $C$ is $\Vol(C) = 2^{-\ell \cdot d}$.
For $0 < x \le 1$ we let $\ceilpowtwo{x}$ be the smallest number larger or equal to $x$ that is realized as the volume of a cell, or in other words $x$ rounded up to a power of $2^d$, $\ceilpowtwo{x} = \min\{2^{- \ell\cdot d} \mid \ell \in \mathbb N_0 \colon 2^{- \ell\cdot d} \ge x \}$.
Note that the cells of a fixed level $\ell$ partition the ground space. We obtain a \emph{geometric ordering} of these cells by following the recursive construction of cells in a breadth-first-search manner, yielding the following lemma.

%We will at several points use a \emph{geometric ordering}, i.e., an enumeration of the vertices that reflects the geometry. It can easily be constructed recursively. We omit the proof.
\begin{lemma}[Geometric ordering]\label{lem:geoorder}
There is an enumeration of the cells $C_1,\ldots,C_{2^{\ell d}}$ of level $\ell$ such that for every cell $C$ of level $\ell' < \ell$ the cells of level $\ell$ contained in $C$ form a consecutive block $C_i,\ldots,C_j$ in the enumeration. 
%Moreover, given $C$ we can compute the first cell $C_i$ and last cell $C_j$ in time $O(1)$ from the coordinates of $C$.
\end{lemma}
%
%We will at several points use a \emph{geometric ordering}, i.e., an enumeration of the vertices that reflects the geometry. The following lemma asserts that this is possible.
%\begin{lemma}[geometric ordering]\label{lem:geoorder}
%Let $s\in \N$. For every $0\leq i \leq s$ let $h_i$ be the regular, axis-parallel grid of side length $2^{-i}$ on $\Space$ passing through the origin, and let $\{c_j^i\}_{j=1}^{2^{di}}$ be the set of cells in this grid. There is an enumeration of the cells in $h_s$, such that for every $0 \leq i\leq s$ and every $1 \leq j \leq 2^{di}$ the $h_s$-cells contained in $c_j^i$ form a consecutive block in the enumeration. Moreover, we can compute the coordinates of the first and last cell in this block in time $O(1)$ from the coordinates of $c_j^i$.
%\end{lemma}
\begin{proof}
We construct the geometric ordering by induction on the level $\ell$. For $\ell=0$ there is only one cell to enumerate, so let $\ell >0$. Given an enumeration $C_1,\ldots,C_{2^{(\ell-1)d}}$ of the cells of level $\ell-1$, we first enumerate all cells of level $\ell$ contained in $C_1$, starting with the cell which is smallest in all $d$ coordinates, and ending with the cell which is largest in all $d$ coordinates. Then we enumerate all cells of level $\ell$ contained in $C_2$ (starting with smallest coordinates, and ending with largest coordinates), and so on. Evidently this gives us a geometric ordering of the cells of level $\ell$. 
%Moreover, if for any cell $C = [x_1, x_1+ 2^{-i}) \times \ldots \times [x_d, x_d+2^{-i})$, the corresponding block of cells of level $\ell$ contained in $C$ starts with the cell $[x_1, x_1+2^{-\ell}) \times \ldots \times [x_d, x_d+2^{-\ell})$ and ends with the cell $[x_1+2^{-i}-2^{-\ell}, x_1+2^{-i})\times \ldots \times [x_d + 2^{-i}-2^{-\ell}, x_d+2^{-i})$, which are computable in constant time.
\end{proof}

\subsection{Basic Properties of GIRGs} \label{subsec:basicproperties}

In this section, we list some basic properties about GIRGs which we mentioned already in Section~\ref{subsec:model} and which repeatedly occur in our proofs. In particular we consider the expected degree of a vertex and the marginal probability that an edge between two vertices with given weights is present. The proofs of all statements follow from more general considerations and can be found in \cite{bringmann2015generalGIRG}. Let us start by calculating the partial weight sums $\W_{\le w}$ and $\W_{\ge w}$. The values of these sums follow from the assumptions on power-law weights in Section~\ref{subsec:model}.

\begin{lemma}[Lemma~4.1 in \cite{bringmann2015generalGIRG}] \label{lem:totalweight}
The total weight satisfies $\W=\Theta(n)$. Moreover, for all sufficiently small $\eta > 0$, 
\begin{enumerate}[(i)]
\item $\W_{\ge w} = O( n w^{2-\beta+\eta})$ for all $w \ge \wmin$,
\item $\W_{\ge w} = \Omega( n w^{2-\beta-\eta})$ for all $\wmin \le w \le \barw$,
\item $\W_{\le w} = O(n)$ for all $w$, and
\item $\W_{\le w} = \Omega(n)$ for all $w=\omega(1)$.
\end{enumerate}
\end{lemma}

Next we consider the marginal edge probability of two vertices $u$, $v$ with weights $\w{u}$, $\w{v}$. In GIRGs, this probability is essentially the same as in Chung-Lu random graphs. Furthermore, the marginal probability does not change by more than a constant factor if we fix the position $\x{u}$ or $\x{v}$ (but not both!). Moreover, conditioned on a fixed position $\x{v} \in \Space$, all edges $\{u,v\}$ are independently present. This is a central feature of our model.

\begin{lemma}[Lemma~4.2 in \cite{bringmann2015generalGIRG}] \label{lem:marginal}
  Fix $u \in [n]$ and $\x u \in \Space$. All edges $\{u,v\}$, $u \ne v$, are independently present with probability
  \begin{equation*}
  \Pr[u \sim v \mid \x{u}] = \Theta(\Pr[u \sim v]) =   \Theta\left(\min\left\{1,\frac{\w{u} \w{v}}{\W} \right\}\right).
  \end{equation*}
\end{lemma}

The following statement shows that the expected degree of a vertex is of the same order as the weight of the vertex, thus we can interpret a given weight sequence $\w{}$ as a sequence of expected degrees.

\begin{lemma}[Lemma~4.3 in \cite{bringmann2015generalGIRG}] \label{lem:expecteddegree}
For any $v \in [n]$ we have $\Ex[\deg(v)]=\Theta(\w{v})$.
\end{lemma}

As the expected degree of a vertex is roughly the same as its weight, it is no surprise that whp the degree of all vertices with weight sufficiently large is concentrated around the expected value. The following lemma gives a precise statement.

\begin{lemma}[Lemma~4.4 in \cite{bringmann2015generalGIRG}] \label{lem:largevertices}
The following properties hold whp for all $v \in [n]$.
\begin{enumerate}[(i)]
\item $\deg(v) = O(\w{v} + \log^2 n)$.
\item If $\w{v} = \omega(\log^2 n)$, then $\deg(v)= (1+o(1))\Ex[\deg(v)]= \Theta(\w{v})$.
\item $\sum_{v \in V_{\ge w}} \deg(v) = \Theta(\W_{\ge w})$ for all $w=\omega(\log^2 n)$.
\end{enumerate}
\end{lemma}

We conclude this section by proving that if we sample the weights randomly from an appropriate distribution, then almost surely the resulting weights satisfy our conditions on power-law weights. In particular, the following lemma implies that all results of this paper for weights satisfying (PL1) and (PL2) also hold almost surely in a model of sampled weights.

\begin{lemma}[Lemma~4.5 in \cite{bringmann2015generalGIRG}] \label{lem:sampleweights}
Let  $\wmin=\Theta(1)$, let $\varepsilon>0$, and let $F=F_n: \R \rightarrow [0,1]$ be non-decreasing such that $F(z)=0$ for all $z \le \wmin$, and $F(z)=1-\Theta(z^{1-\beta})$ for all $z \in [\wmin,n^{1/(\beta-1-\varepsilon)}]$. Suppose that for every vertex $v \in [n]$, we choose the weight $\w{v}$ independently according to the cumulative probability distribution $F(.)$. Then for all $\eta=\eta(n)=\omega(\log\log n / \log n)$ , with probability $1-n^{-\Omega(\eta)}$, the resulting weight vector $\w{}$ satisfies the power-law conditions (PL1) and (PL2) with $\barw = (n/\log^2 n)^{1/(\beta-1)}$. 
\end{lemma}

\section{Sampling Algorithm} \label{sec:sampling}

In this section we show that GIRGs can be sampled in expected time $O(n)$. The running time depends exponentially on the fixed dimension $d$. In addition to our model assumptions, in this section we require that (1) edge probabilities $p_{uv}$ can be computed in constant time (given any vertices $u,v$ and positions $\x u, \x v$) and (2) we know an explicit constant $c > 0$ such that if $\alpha < \infty$ we have
$$ p_{uv} \le \min\bigg\{c \frac 1{\|\x u - \x v\|^{\alpha d}} \cdot \Big(\frac{\w u \w v}\W \Big)^\alpha, 1\bigg\}, $$
and if $\alpha = \infty$ we have
$$ p_{uv} \le \begin{cases} 1, & \text{if } \|\x u - \x v\| < c \big(\tfrac{\w u \w v}\W\big)^{1/d}, \\ 0, & \text{otherwise.} \end{cases}  $$
Note that existence of $c$ follows from our model assumptions. 
In the remainder of this section we introduce building blocks of our algorithm (Section~\ref{subsec:samplingBB}) and present our  algorithm (Section~\ref{subsec:samplingALG}) and its analysis (Section~\ref{subsec:samplingANA}). Thereby, we always assume $\alpha<\infty$. In the last part of this chapter (Section~\ref{subsec:samplingALPHA}), we show how the sampling algorithm can be adapted to the case $\alpha=\infty$.

\subsection{Building Blocks}
\label{subsec:samplingBB}

\paragraph*{Data structures:}
Recall the definition of cells from Section~\ref{subsec:cells}. We first build a basic data structure on a set of points $P$ that allows to access the points in a given cell $C$ (of volume at least~$\nu$) in constant time. 

\begin{lemma} \label{lem:samplingPointLocation}
  Given a set of points $P$ and $0 < \nu \le 1$, in time $O(|P|+1/\nu)$ we can construct a data structure $\mathcal D_\nu(P)$ supporting the following queries in time $O(1)$:
  \begin{itemize}
    \item given a cell $C$ of volume at least $\nu$, return $|C \cap P|$,
    \item given a cell $C$ of volume at least $\nu$ and a number $k$, return the $k$-th point in $C \cap P$ (in a fixed ordering of $C \cap P$ depending only on $P$ and $\nu$).
  \end{itemize}
\end{lemma} 
\begin{proof}
  Let $\mu = \ceilpowtwo{\nu} = 2^{-\ell \cdot d}$, so that $\nu \le \mu \le O(\nu)$. Following the recursive construction of cells, we can determine a geometric ordering of the cells of volume $\mu$ as in Lemma~\ref{lem:geoorder} in time $O(1/\mu) = O(1/\nu)$; say $C_1,\ldots,C_{1/\mu}$ are the cells of volume $\mu$ in the geometric ordering. We store this ordering by storing a pointer from each cell $C_i = (\ell,x_1,\ldots,x_d)$ to its successor $C_{i+1} = (\ell,x_1',\ldots,x_d')$, which allows to scan the cells $C_1,\ldots,C_{1/\mu}$ in linear time.
  For any point $x \in P$, using the floor function we can determine in time $O(1)$ the cell $(\ell,x_1,\ldots,x_d)$ of volume $\mu$ that $x$ belongs to. This allows to determine the numbers $|C_i \cap P|$ for all $i$ in time $O(|P| + 1/\nu)$. We also compute each prefix sum $s_i := \sum_{j < i} |C_j \cap P|$ and store it at cell $C_i = (\ell,x_1,\ldots,x_d)$. Using an array $A[.]$ of size $|P|$, we store (a pointer to) the $k$-th point in $C_i \cap P$ at position $A[s_i + k]$. Note that this preprocessing can be performed in time $O(|P|+1/\nu)$.
  
  A given cell $C$ of volume at least $\nu$ may consist of several cells of volume $\mu$. By Lemma~\ref{lem:geoorder}, these cells form a contiguous subsequence $C_i,C_{i+1},\ldots,C_{j-1},C_j$ of $C_1,\ldots,C_{1/\mu}$, so that the points $C \cap P$ form a contiguous subsequence of $A$. For constant access time, we store for each cell $C$ of volume at least $\nu$ the indices $s_C,e_C$ of the first and last point of $C \cap P$ in $A$. Then $|C \cap P| = e_C - s_C + 1$ and the $k$-th point in $C \cap P$ is stored at $A[s_C+k]$. Thus, both queries can be answered in constant time. Note that the ordering $A[.]$ of the points in $C \cap P$ is a mix of the geometric ordering of cells of volume $\mu$ and the given ordering of $P$ within a cell of volume $\mu$, in particular this ordering indeed only depends on $P$ and $\nu$.
\end{proof}

Next we construct a partitioning of $\Space \times \Space$ into products of cells $A_i \times B_i$. This partitioning allows to split the problem of sampling the edges of a GIRG into one problem for each $A_i \times B_i$, which is beneficial, since each product $A_i \times B_i$ has one of two easy types.
For any $A,B \subseteq \Space$ we denote the distance of $A$ and $B$ by $d(A,B) = \inf_{a \in A, b \in B} \|a-b\|$. 

\begin{lemma} \label{lem:samplingPartition}
  Let $0 < \nu \le 1$. In time $O(1/\nu)$ we can construct a set\\$\mathcal P_\nu = \{(A_1,B_1),\ldots,(A_s,B_s)\}$ such that
  \begin{enumerate}[(1)]
    \item $A_i,B_i$ are cells with $\Vol(A_i) = \Vol(B_i) \ge \nu$,
    \item for all $i$, either $d(A_i,B_i) = 0$ and $\Vol(A_i) = \lceil \nu \rceil_{2^d}$ (type I) or $d(A_i,B_i) \ge \Vol(A_i)^{1/d}$ (type~II),
    \item the sets $A_i \times B_i$ partition $\Space \times \Space$,
    \item $s = O(1/\nu)$.
  \end{enumerate}
\end{lemma}
\begin{proof}
  Note that for cells $A,B$ of equal volume we have $d(A,B) = 0$ if and only if either $A=B$ or (the boundaries of) $A$ and $B$ touch. For a cell $C$ of level $\ell$ we let $\parent(C)$ be its \emph{parent}, i.e., the unique cell of level $\ell-1$ that $C$ is contained in.
  Let $\mu = \ceilpowtwo{\nu}$. We define $\mathcal P_\nu$ as follows. For any pair of cells $(A,B)$ with $\Vol(A) = \Vol(B) \ge \nu$, we add $(A,B)$ to $\mathcal P_\nu$ if either (i) $\Vol(A) = \Vol(B) = \mu$ and $d(A,B) = 0$, or (ii) $d(A,B) > 0$ and $d(\parent(A),\parent(B)) = 0$. 
  
  Property (1) follows by definition. 
  Regarding property (2), the pairs $(A,B)$ added in case~(i) are clearly of type I. Observe that two cells $A,B$ of equal volume that are not equal or touching have distance at least the sidelength of $A$, which is $\Vol(A)^{1/d}$. Thus, in case (ii) the lower bound $d(A,B) > 0$ implies $d(A,B) \ge \Vol(A)^{1/d}$, so that $(A,B)$ is of type II. 
  
  For property (3), consider $(x,y) \in \Space \times \Space$ and let $A,B$ be the cells of volume $\mu$ containing~$x,y$. 
  Let $A^{(0)} := A$ and $A^{(i)} := \parent(A^{(i-1)})$ for any $i\ge 1$, until $A^{(k)} = \Space$. Similarly, define $B = B^{(0)} \subset \ldots \subset B^{(k)} = \Space$ and note that $\Vol(A^{(i)}) = \Vol(B^{(i)})$. Observe that each set $A^{(i)} \times B^{(i)}$ contains $(x,y)$. Moreover, any set $A' \times B'$, where $A',B'$ are cells with $\Vol(A') = \Vol(B')$ and $(x,y) \in A' \times B'$, is of the form $A^{(i)} \times B^{(i)}$. Thus, to show that $\mathcal P_\nu$ partitions $\Space \times \Space$ we need to show that it contains exactly one of the pairs $(A^{(i)},B^{(i)})$ (for any $x,y$). To show this, we use the monotonicity $d(A^{(i)},B^{(i)}) \ge d(A^{(i+1)},B^{(i+1)})$ and consider two cases. If $d(A,B) = 0$ then we add $(A,B)$ to $\mathcal P_\nu$ in case (i), and we add no further $(A^{(i)},B^{(i)})$, since $d(A^{(i)},B^{(i)}) = 0$ for all~$i$. If $d(A,B) > 0$ then since $d(A^{(k)},B^{(k)}) = d(\Space,\Space) = 0$ there is a unique index $0 \le i < k$ with $d(A^{(i)},B^{(i)}) > 0$ and $d(A^{(i+1)},B^{(i+1)}) = 0$. Then we add $(A^{(i)},B^{(i)})$ in case (ii) and no further $(A^{(j)},B^{(j)})$. This proves property (3).
  
  Property (4) follows from the running time bound of $O(1/\nu)$, which we show in the following. Note that we can enumerate all $1/\mu = O(1/\nu)$ cells of volume $\mu$, and all of the at most $3^d = O(1)$ touching cells of the same volume, in time $O(1/\nu)$, proving the running time bound for case (i). Moreover, we can enumerate all $2^{\ell \cdot d}$ cells $C$ in level $\ell$, together with all of the at most $3^d = O(1)$ touching cells $C'$ in the same level. Then we can enumerate all $2^d = O(1)$ cells $A$ that have $C$ as parent as well as all $O(1)$ cells $B$ that have $C'$ as parent. This enumerates (a superset of) all possibilities of case (ii). Summing the running time $O(2^{\ell\cdot d})$ over all levels $\ell$ with volume $2^{-\ell \cdot d} \ge \nu$ yields a total running time of $O(1/\nu)$.
\end{proof}

\paragraph*{Weight layers:}
We set $w_0 := \wmin$ and $w_i := 2 w_{i-1}$ for $i \ge 1$. This splits the vertex set $V = [n]$ into \emph{weight layers} $V_i := \{v \in V \mid w_{i-1} \le v < w_i\}$ for $1 \le i \le L$ with $L = O(\log n)$. 
We write $V_i^C$ for the restriction of weight layer $V_i$ to cell $C$, $V_i^C := \{v \in V_i \mid \x v \in C\}$.

\paragraph*{Geometric random variates:}
For $0 < p \le 1$ we write $\Geo(p)$ for a geometric random variable, taking value $i \ge 1$ with probability $p (1-p)^{i-1}$. $\Geo(p)$ can be sampled in constant time using the simple formula $\big\lceil \frac{\log(R)}{\log(1-p)} \big\rceil$, where $R$ is chosen uniformly at random in $(0,1)$, see~\cite{devroye_nonuniform_1986}\footnote{To evaluate this formula exactly in time $O(1)$ we need to assume the RealRAM model of computation. However, also on a bounded precision machine like the WordRAM $\Geo(p)$ can be sampled in expected time $O(1)$~\cite{bringmann_exact_2013}.}.

\subsection{The Algorithm}
\label{subsec:samplingALG}

\begin{algorithm}
\caption{Sampling algorithm for GIRGs in expected time $O(n)$}\label{alg:sampling}
\begin{algorithmic}[1]
  \State $E := \emptyset$
  \State sample the positions $\x v$, $v \in V$, and determine the weight layers $V_i$
  \ForAll{$1 \le i \le L$}
    build data structure $\mathcal D_{\nu(i)}(\{\x v \mid v \in V_i\})$ with $\nu(i) := \frac{w_i w_0}\W$
  \EndFor
  \ForAll{$1 \le i \le j \le L$}
    \State construct partitioning $\mathcal P_{\nu(i,j)}$ with $\nu(i,j) := \frac{w_i w_j}\W$
    \ForAll{$(A,B) \in \mathcal P_{\nu(i,j)}$ of type I}
      \ForAll{$u \in V_i^A$ and $v \in V_j^B$}  with probability $p_{uv}$ add edge $\{u,v\}$ to $E$
      \EndFor
    \EndFor
    \ForAll{$(A,B) \in \mathcal P_{\nu(i,j)}$ of type II}
      \State $\bar p := \min\big\{c \cdot \frac 1{d(A,B)^{\alpha d}} \cdot \big(\frac{w_i w_j}\W\big)^\alpha, 1\big\}$
      \State $r := \Geo(\bar p)$
      \While{$r \le |V_i^A| \cdot |V_j^B|$}
        \State determine the $r$-th pair $(u,v)$ in $V_i^A\times V_j^B$
        \State with probability $p_{uv} / \bar p$ add edge $\{u,v\}$ to $E$
        \State $r := r + \Geo(\bar p)$
      \EndWhile
    \EndFor
    \If{$i=j$}
      remove all edges with $u>v$ sampled in this iteration
    \EndIf
  \EndFor
\end{algorithmic}
\end{algorithm}

Given the model parameters, our Algorithm~\ref{alg:sampling} samples the edge set $E$ of a GIRG. To this end, we first sample all vertex positions $\x v$ uniformly at random in $\Space$. Given weights $w_1,\ldots,w_n$ we can determine the weight layers $V_i$ in linear time (we may use counting sort or bucket sort since there are only $L = O(\log n)$ layers). Then we build the data structure from Lemma~\ref{lem:samplingPointLocation} for the points in $V_i$ setting $\nu = \nu(i) = \frac{w_i w_0}\W$, i.e., we build $\mathcal D_{\nu(i)}(\{\x v \mid v \in V_i\})$ for each $i$. In the following, for each pair of weight layers $V_i,V_j$ we sample the edges between $V_i$ and $V_j$. To this end, we construct the partitioning $\mathcal P_{\nu(i,j)}$ from Lemma~\ref{lem:samplingPartition} with $\nu(i,j) = \frac{w_i w_j}\W$. Since $\mathcal P_{\nu(i,j)}$ partitions $\Space \times \Space$, every pair of vertices $u \in V_i, v \in V_j$ satisfies $\x u \in A, \x v \in B$ for exactly one $(A,B) \in \mathcal P_{\nu(i,j)}$. Thus, we can iterate over all $(A,B) \in \mathcal P_{\nu(i,j)}$ and sample the edges between $V_i^A$ and $V_j^B$. 

If $(A,B)$ is of type I, then we simply iterate over all vertices $u \in V_i^A$ and $v \in V_j^B$ and add the edge $\{u,v\}$ with probability $p_{uv}$; this is the trivial sampling algorithm. Note that we can efficiently enumerate $V_i^A$ and $V_j^B$ using the data structure $\mathcal D_{\nu(i)}(\{\x v \mid v \in V_i\})$ that we constructed above. 

If $(A,B)$ is of type II, then the distance $\|x-y\|$ of any two points $x \in A, y \in B$ satisfies $d(A,B) \le \|x-y\| \le d(A,B) + \Vol(A)^{1/d} + \Vol(B)^{1/d} \le 3 d(A,B)$, by the definition of type~II. Thus, $\bar p = \min\big\{c \cdot \frac 1{d(A,B)^{\alpha d}} \cdot \big(\frac{w_i w_j}\W\big)^\alpha, 1\big\}$ is an upper bound on the edge probability~$p_{uv}$ for any $u \in V_i^A, v \in V_j^B$, and it is a good upper bound since $d(A,B)$ is within a constant factor of $\|\x u - \x v\|$ and $w_i,w_j$ are within constant factors of $\w u, \w v$. Now we first sample the set of edges $\bar E$ between $V_i^A$ and $V_j^B$ that we would obtain if all edge probabilities were equal to $\bar p$, i.e., any $(u,v) \in V_i^A \times V_j^B$ is in $\bar E$ independently with probability $\bar p$. From this set $\bar E$, we can then generate the set of edges with respect to the true edge probabilities~$p_{uv}$ by throwing a coin for each $\{u,v\} \in \bar E$ and letting it survive with probability $p_{uv} / \bar p$. Then in total we choose a pair $(u,v)$ as an edge in $E$ with probability $\bar p \cdot (p_{uv}/\bar p) = p_{uv}$, proving that we sample from the correct distribution. Note that here we used $p_{uv} \le \bar p$. It is left to show how to sample the ``approximate'' edge set $\bar E$. First note that the data structure $\mathcal D_\nu(\{\x v \mid v \in V_i\})$ defines an ordering on $V_i^A$, and we can determine the $\ell$-th element in this ordering in constant time, similarly for $V_j^B$. Using the lexicographic ordering, we obtain an ordering on $V_i^A\times V_j^B$ for which we can again determine the $\ell$-th element in constant time. In this ordering, the first pair $(u,v) \in V_i^A\times V_j^B$ that is in $\bar E$ is geometrically distributed, according to $\Geo(\bar p)$. Since geometric random variates can be generated in constant time, we can efficiently generate $\bar E$, specifically in time $O(1 + |\bar E|)$.

Finally, the case $i=j$ is special. With the algorithm described above, for any $u,v \in V_i$ we sample whether they form an edge twice, once for $\x u \in A, \x v \in B$ (for some $(A,B) \in \mathcal P_{\nu(i,j)}$) and once for $\x v \in A', \x u \in B'$ (for some $(A',B') \in \mathcal P_{\nu(i,j)}$). To fix this issue, in the case $i=j$ we only accept a sampled edge $(u,v) \in V_i^A \times V_j^B$ if $u < v$; then only one way of sampling edge $\{u,v\}$ remains. This changes the expected running time only by a constant factor.

%We now describe our algorithm and analyze its running time. First, we sample all vertex positions $\x v$ uniformly at random in $\Space$, and we determine the weight levels $V_i$, both of which can be performed in linear time. Then we build the data structure from Lemma~\ref{lem:samplingPointLocation} for the points in $V_i$ with volume at least $\nu = \frac{w_i w_0}n$. As $w_0 = \wmin = \Omega(1)$, this step runs in time $\sum_{i=1}^L O(|V_i| + n/w_i)$, which is $O(n)$ since the sets $V_i$ partition $V$ and $w_i$ increases exponentially. Then for each pair of weight layers $V_i,V_j$ we sample the edges from $V_i$ to $V_j$. To this end, we construct the partitioning $\mathcal P_\nu$ from Lemma~\ref{lem:samplingPartition} with $\nu = \frac{w_i w_j}n$. The total time spend for this step is $\sum_{i \le j} O(n/(w_i w_j))$, which is $O(n)$ since $w_i$ increases exponentially. 

\subsection{Analysis}
\label{subsec:samplingANA}

Correctness of our algorithm follows immediately from the above explanations. In the following we show that Algorithm~\ref{alg:sampling} runs in expected linear time. This is clear for lines 1-2. For line~3, since building the data structure from Lemma~\ref{lem:samplingPointLocation} takes time $O(|P| + 1/\nu)$, it takes total time $\sum_{i=1}^L O\big(|V_i| + \W/(w_i w_0)\big)$. Clearly, the first summand $|V_i|$ sums up to $n$. Using $w_0 = \wmin = \Omega(1)$, $\W = O(n)$, and that $w_i$ grows exponentially with~$i$, implying $\sum_i 1/w_i = O(1)$, also the second summand sums up to $O(n)$. For line 5, all invocations in total take time $O\big( \sum_{i,j} \W/(w_i w_j)\big)$, which is $O(n)$, since again $\W = O(n)$ and $\sum_i 1/w_i = O(1)$.
We claim that for any weight layers $V_i,V_j$ the expected running time we spend on any $(A,B) \in \mathcal P_{\nu(i,j)}$ is $O(1 + \Ex[|E^{A,B}_{i,j}|])$, where $E^{A,B}_{i,j}$ is the set of edges in $V_i^A \times V_j^B$. Summing up the first summand $O(1)$ over all $(A,B) \in \mathcal P_{\nu(i,j)}$ sums up to $1/\nu(i,j) = \W/(w_i w_j)$. As we have seen above, this sums up to $O(n)$ over all $i,j$. Summing up the second summand $O(\Ex[|E^{A,B}_{i,j}|])$ over all $(A,B) \in \mathcal P_{\nu(i,j)}$ and weight layers $V_i,V_j$ yields the total expected number of edges $O(\Ex[|E|])$, which is $O(n)$, since the average weight $\W / n = O(1)$ and thus the expected average degree is constant.

It is left to prove the claim that for any weight layers $V_i,V_j$ the expected time spent on $(A,B) \in \mathcal P_{\nu(i,j)}$ is $O(1 + \Ex[|E^{A,B}_{i,j}|])$. If $(A,B)$ is of type I, then any pair of vertices $(u,v) \in V_i^A\times V_j^B$ has probability $\Theta(1)$ to form an edge: Since the volume of $A$ and $B$ is $w_i w_j/\W$, their diameter is $(w_i w_j/\W)^{1/d}$ and we obtain $\|\x u - \x v\| \le (w_i w_j/\W)^{1/d} = O((\w u \w v / \W)^{1/d})$, which yields $p_{uv} = \Theta\big(\min\big\{\big(\frac{\w u \w v}{\|\x u - \x v\|^d \W}\big)^\alpha,1\big\}\big) = \Theta(1)$. 
As we spend time $O(1)$ for any $(u,v) \in V_i^A\times V_j^B$, we stay in the desired running time bound $O(\Ex[|E^{A,B}_{i,j}|])$. 

If $(A,B)$ is of type II, we first sample edges $\bar E$ with respect to the larger edge probability~$\bar p$, and then for each edge $e \in \bar E$ sample whether it belongs to $E$. This takes total time $O(1 + |\bar E|)$. Note that any edge $e \in \bar E$ has constant probability $p_{uv} / \bar p = \Theta(1)$ to survive: It follows from $\w u = \Theta(w_i), \w v = \Theta(w_j)$, and $\|\x u - \x v\| = \Theta(d(A,B))$ that $p_{uv} = \Theta(\bar p)$. Hence, we obtain $\Ex[|\bar E|] = O(\Ex[|E^{A,B}_{i,j}|])$, and the running time $O(1 + |\bar E|)$ is therefore in expectation bounded by $O(1 + \Ex[|E^{A,B}_{i,j}|])$. This finishes the proof of the claim.

\subsection{Sampling GIRGs in the Case $\alpha = \infty$}
\label{subsec:samplingALPHA}

For $\alpha = \infty$, edges only exist between vertices in distance $\|\x u - \x v\| < c (\w u \w v / \W)^{1/d}$. We change Algorithm~\ref{alg:sampling} by setting $\nu(i,j) = \max\{1,c\}^d \cdot w_i w_j / \W$. Then for any $u \in V_i, v \in V_j$ and $(A,B) \in \mathcal P_{\nu(i,j)}$ of type II we have $d(A,B) \ge \Vol(A)^{1/d} \ge \nu^{1/d} \ge c (\w u \w v / \W)^{1/d}$, so there are no edges between $V_i^A$ and $V_j^B$ for type II. This allows to simplify the algorithm by completely ignoring type~II pairs; the rest of the algorithm stays unchanged. 

Additionally, we have to slightly change the running time analysis, since it no longer holds that all pairs of vertices $(u,v) \in V_i^A\times V_j^B$ satisfy $p_{uv} = \Theta(1)$. However, a variant of this property still holds:  If we only uncovered that $\x u \in A$ and $\x v \in B$, but not yet where exactly in $A,B$ they lie, then the marginal probability of $(u,v)$ forming an edge is $\Theta(1)$, since for any $\eps > 0$ a constant fraction of all pairs of points in $A \times B$ are within distance $\eps (w_{i-1} w_{j-1} / \W)^{1/d}$, guaranteeing edge probability $\Theta(1)$ for sufficiently small $\eps$. This again allows to check all pairs of vertices in $V_i^A\times V_j^B$ whether they form an edge, which yields expected linear running time.

\section{Clustering}\label{sec:clustering}

In this section we give a proof of Theorem~\ref{thm:clustering}. We start by stating the formal definition of the clustering coefficient. 

\begin{definition} \label{def:clustering}
In a graph $G=(V,E)$ the clustering coefficient of a vertex $v\in V$ is defined as
\[
\cc(v) := \cc_G(v) := \begin{cases} \#\big\{\{u,u'\} \subseteq \Gamma(v) \,\mid\, u \sim u'\big\} \big/ \binom{\deg v}{2}, & \text{if} \deg(v) \geq 2, \\ 0 ,& \text{otherwise,}\end{cases}
\]
and the (mean) clustering coefficient of $G$ is 
\[
\cc(G) := \frac{\sum_{v\in V}\cc(v)}{|V|}.
\]
\end{definition}

For the proof of Theorem~\ref{thm:clustering}, we need Le Cam's theorem which allows us to bound the total variation
distance of a binomial distribution to a Poisson distribution with the same
mean.
\begin{theorem}[Le Cam, Proposition 1 in~\cite{le1960approximation}]\label{thm:lecam}
  Suppose that $X_1,\ldots,X_n$ are independent Bernoulli random variables such that
  $\Pr[X_i=1]=p_i$ for $i\in [n]$, $\lambda_n = \sum_{i\in [n]}p_i$ and
  $S_n=\sum_{i\in [n]} X_i$. Then
  \[
    \sum_{k=0}^\infty\left|\Pr[S_n=k] - \frac{\lambda_n^k e^{-\lambda_n}}{k!}\right| < 2\sum_{i=1}^n p_i^2.
  \]
  In particular, if $\lambda_n = \Theta(1)$ and $\max_{i \in [n]}p_i = o(1)$, then $\Pr[S_n = k] = \Theta(1)$ for $k= O(1)$.
\end{theorem}

Furthermore, for showing concentration of the clustering coefficient we need a concentration inequality which bounds large deviations taking into account some bad event $\mathcal{B}$. 
We use the following theorem.
\begin{theorem}[Theorem~3.3 in \cite{bringmann2015generalGIRG}]
\label{thm:concentration}
Let $X_1,\ldots,X_m$ be independent random variables over $\Omega_1, \ldots, \Omega_m$. Let $X = (X_1,\ldots,X_m)$, $\Omega = \prod_{k=1}^m \Omega_k$ and let $f\colon \Omega \to \mathbb{R}$ be measurable with $0 \le f(\omega) \le M$ for all $\omega \in \Omega$. Let $\mathcal{B} \subseteq \Omega$ such that %$X$ is \emph{strongly $c$-Lipschitz with bad event $\mathcal{B}$}, i.e., 
for some $c > 0$ and for all $\omega \in \overline{\mathcal{B}}$, $\omega' \in \overline{\mathcal{B}}$ that differ in at most \emph{two} components we have
\[
|f(\omega)-f(\omega')| \le c.
\]
Then for all $t \ge 2 M \Pr[\mathcal B]$
\[\Pr\big[\vert f(X)-\Ex[f(X)] \vert \ge t \big] \le 2e^{-\frac{t^2}{32mc^2}}+(2\tfrac{mM}{c} + 1)\Pr[\mathcal{B}].\]
\end{theorem}

\begin{proof}[Proof of Theorem~\ref{thm:clustering}]
Let $V' := V_{\le n^{1/8}}$ and $G'=G[V']$. 
We first show that for the subgraph $G'$ we have $\Ex[\cc(G')] = \Omega(1)$. %Afterwards we will use the same techniques as in Section~\ref{sec:degreesequence} to prove the high-probability statement for $\cc(G)$. 
Let $w_0 = \Theta(1)$ be a weight such that there are linearly many vertices with weight at most $w_0$. 

Since $\cc(G') = \frac{1}{|V'|}\sum_{v\in V'} \cc_{G'}(v)=\Theta(\frac{1}{n}\sum_{v\in V'} \cc_{G'}(v))$, it suffices to show that a vertex $v$ of weight at most $w_0$ fulfills $\Ex[\cc_{G'}(v)] = \Omega(1)$.  
%Let $r$ be such that $V(r) = 1/n$ (this exists because the function $V$ is surjective). Let $r'$ be the corresponding radius given by Definition~\ref{def:probtriangle}, and note that $V(r') = \Theta(1/n)$. 
For this we consider the set $\barV:=V_{\leq w_0}$ of vertices of weight at most $w_0$. Fix such a vertex $v$ at position $\x{v} \in \Space$, and let $U(v)$ be the ball around $\x{v}$ with radius $cn^{-1/d}$ for a sufficiently small constant $c>0$. Clearly the volume of this ball is $\Theta(n^{-1})$. Thus, the expected number of vertices in $\barV$ with position in $U(v)$ is $\Theta(1)$. Consider the event~$\mathcal{E} = \mathcal{E}(v)$ that the following three properties hold.
\begin{enumerate}[(i)]
\item $v$ has at least two neighbors in $\barV$ with positions in $U(v)$.
\item $v$ does not have neighbors in $\barV$ with positions in $\Space\setminus U(v)$.
\item $v$ does not have neighbors in $[n]\setminus \barV$.
\end{enumerate}

We claim that $\Pr[\mathcal{E}] = \Theta(1)$. For (i), note that the expected number of vertices in $\barV$ with position in $U(v)$ is $\Theta(1)$. Since the position of every vertex is independent, by Le Cam's theorem (Theorem~\ref{thm:lecam}) the probability that there are at least two vertices in $\barV$ with position in $U(v)$ is $\Theta(1)$. Moreover, by (\ref{eq:puv}) and (\ref{eq:puv2}), for each such vertex the probability to connect to $v$ is $\Theta(1)$ (in the case $\alpha=\infty$ this follows because we have chosen $c$ small enough), so (i) holds with probability $\Theta(1)$. For (ii), for any vertex $u \in \barV\setminus\{v\}$, we can bound 
\[
\Pr[v \sim u, \x{v} \in \Space\setminus U(v)] \leq \Pr[v \sim u] = \Theta(1/n).
\]
Hence, by Le Cam's theorem, (ii) holds with probability~$\Theta(1)$, and this probability can only increase if we condition on (i). Finally, for every fixed position $x$, (iii) holds independently of (i) or (ii) with probability $\Theta(1)$, again by Le Cam's theorem. This proves the claim that $\Pr[\mathcal{E}] = \Theta(1)$. 

Conditioned on $\mathcal{E}$, let $v_1$ and $v_2$ be two random neighbors of $v$. Then $\x{v_1}, \x{v_2} \in U(v)$, and $\w{v_1},\w{v_2} \leq w_0$. By the triangle inequality we obtain $\|\x{v_1}-\x{v_2}\| \le 2c n^{-1/d}$. For $c$ sufficiently small, we deduce from (\ref{eq:puv}) and (\ref{eq:puv2}) that $v_1 \sim v_2$ holds with probability $\Theta(1)$. Thus we have shown that $\Ex[\cc_{G'}(v)\mid \mathcal{E}(v)] = \Omega(1)$ for all $v\in \barV$. Since $\Pr[\mathcal{E}(v)] = \Theta(1)$, this proves $\Ex[\cc_{G'}(v)] = \Omega(1)$ for all $v\in \barV$, which implies $\Ex[\cc(G')]=\Omega(1)$. 

Next we show that $\cc(G')$ is concentrated around its expected value. We aim to do this via an Azuma-type inequality with error event, as given in Theorem~\ref{thm:concentration}. Note that in our graph model, we apply two different randomized processes to create the geometric graph. First, for every vertex $v$ we choose $\x{v}$ independently at random. Afterwards, every edge is present with some probability $p_{uv}$. Recall that we can apply the concentration bound only if all random variables are independent, which is not the case so far.

The $n$ random variables $\x{1},\ldots, \x{n}$ define the vertex set and the edge probabilities $p_{uv}$. We introduce a second set of $n-1$ independent random variables. For every $u \in \{2,\ldots,n\}$ we let $Y_u := (Y_u^1, \ldots, Y_u^{u-1})$, where every $Y_u^v$ is independently and uniformly at random from $[0,1]$. Then for $v<u$, we include the edge $\{u,v\}$ in the graph if and only if
\[p_{u v} > Y_u^v.\]
We observe that indeed this implies $\Pr[u \sim v \mid \x{u},\x{v}] = p_{u v}(\x u, \x v)$ as desired. Furthermore, the $2n-1$ random variables $\x{1},\ldots,\x{n},Y_2, \ldots Y_n$ are independent and define a probability space $\Omega$. Then $G$, $G'$ and $\cc(G')$ are all random variables on $\Omega$. Consider the following bad event:
\begin{equation} \label{eq:badevent}
\mathcal{B} := \{\omega \in \Omega: \text{ the maximum degree in  }G'(\omega) \text{ is at least }n^{1/4}\}.
\end{equation}
We observe that $\Pr[\mathcal{B}]=n^{-\omega(1)}$, since whp every vertex $v \in V'$ has degree at most $O(\w{v} + \log^2 n) = o(n^{1/4})$ by Lemma~\ref{lem:largevertices}. Let $\omega,\omega' \in \overline{\mathcal{B}}$ such that they differ in at most two coordinates. We observe that changing one coordinate $\x{i}$ or $Y_i$ can influence only the local clustering coefficients of $i$ itself and of the vertices which are neighbors of $i$ either before or after the coordinate change. Unless $\mathcal{B}$ holds, every vertex in $G'$ has degree at most $n^{1/4}$ and therefore every coordinate of the probability space has effect at most $2n^{1/4} / n$ onto $\cc(G')$. Thus, we obtain $|\cc(G'(\omega))-\cc(G'(\omega'))| \le 4n^{-3/4}$. We apply Theorem~\ref{thm:concentration} with $t=n^{-1/8}$ and $c := 4n^{-3/4}$ and deduce
\[\Pr\left[|\cc(G')-\Ex[\cc(G')]| \ge t \right] \le 2e^{-\Omega(t^2/n^{1-3/2})} + n^{O(1)}\Pr[\mathcal{B}]  = n^{-\omega(1)},\]
where we used $\frac{t^2}{n^{1-3/2}}=n^{1/4}$. Hence, we have $\cc(G')=(1+o(1))\Ex[\cc(G')]=\Omega(1)$ whp.

In order to compare $\cc(G)$ with $\cc(G')$, we observe that every additional edge $e=\{u,v\}$ which we add to $G'$ can decrease only $\cc(u)$ and $\cc(v)$, both by at most one. Thus,
\[\cc(G) \ge \frac{|V'|}{n}\cc(G') - \frac{2}{n}\sum_{v \in V \setminus V'} \deg(v).\]
By Lemma~\ref{lem:largevertices}, 
$\frac{2}{n}\sum_{v \in V \setminus V'} \deg(v)=\Theta(n^{-1}\W_{> n^{1/8}})=o(1)$
whp. Together with $|V'|=\Theta(n)$, this concludes the argument and proves that $\cc(G)=\Omega(1)$ whp.
%This proves that $\Ex[\cc(G)] = \Omega(1)$. For the high probability statement, split $\barV$ into two sets $V_1$ and $V_2$ of equal size. Moreover, we remove all vertices in $V_1$ that have neighbors in $V_1$. Afterwards, whp $V_1$ will still have constant size. Then we go through the vertices in $V_1$, and replace the set $\barV$ by $V_2$. More precisely, for a vertex $v_k$ in $V_1$, let $\mathcal{E}_k$ be the event that exactly two vertices $v_i,v_j \in V_2$ have distance at most $r$ from $v_k$, that $v_i,v_j$ are adjacent to $v_k$, and that $v_k$ has no other neighbors in $V_2$ and no neighbors in $[n]\setminus \barV$. Then the events $\mathcal{E}_k$ are independent of each other, and they still have probability $\Theta(1)$. Therefore, by the Chernoff-Hoeffding bound, a linear number of them will occur with probability $1- e^{-\Omega(n)}$. Hence, $\cc(G) = \Omega(1)$ with probability $1- e^{-\Omega(n)}$.
\end{proof}

\section{Stability of the Giant, Entropy, and Compression Algorithm}
\label{sec:entropy}

In this section we prove Theorems~\ref{thm:stability} and~\ref{thm:entropy}. More precisely, we show that whp the graph (and its giant) has separators of sublinear size, and we make use of these small separators to devise a compression algorithm that can store the graph using a linear number of bits in expectation. Note that the compression maintains only the graph up to isomorphism, not the underlying geometry. The main idea is to enumerate the vertices in an ordering that reflects the geometry, and then storing for each vertex $i$ the differences $i-j$ for all neighbors $j$ of $i$. We start with a technical lemma that gives the number of edges intersecting an axis-parallel, regular grid. (For $\gamma > 0$ with $1/\gamma \in \N$, the axis-parallel, regular grid with side length $\gamma$ is the union of all $d-1$-dimensional hyperplanes that are orthogonal to an axis and that are in distance $k\gamma$ from the origin for a $k\in \Z$.) Both the existence of small separators and the efficiency of the compression algorithm follow easily from that formula.

\begin{lemma}\label{lem:grid}
Let $\eta > 0$. Let $1 \leq \mu \leq n^{1/d}$ be an integer, and consider an axis-parallel, regular grid with side length $1/\mu$ on $\Space$. Then the expected number of edges intersected by the grid is at most $O(n\cdot (n/\mu^d)^{2-\beta+\eta} + (n^{2-\alpha}\mu^{d(\alpha -1)}+n^{1-1/d}\mu)(1+\log(n/\mu^d)))$.
\end{lemma}

We defer the proof of Lemma~\ref{lem:grid} to the end of this chapter. For the regime $2<\beta<3$, we immediately obtain from the statement that there is a sublinear set of vertices that disconnects the giant component.

\begin{proof}[Proof of Theorem~\ref{thm:stability}]
By Lemma~\ref{lem:grid} for $\mu = 2$, there are $m = O(n^{\max\{3-\beta,2-\alpha,1-1/d\}+\eta})$ edges intersecting a grid of side length $1/2$ in expectation, and two hyperplanes of this grid suffice to split $\Space$ into two halves. Whp there are $\Omega(n)$ vertices in each grid cell, and whp the weights of the vertices in each half satisfy a power law. If $2<\beta<3$ then whp each halfspace gives rise to a giant component of linear size, this follows from more general considerations in~\cite{bringmann2015generalGIRG}.  Hence, whp the two hyperplanes split the giant of $G$ into two parts of linear size, although almost surely they only intersect $n^{1-\Omega(1)}$ edges. Finally, since the bound $m = O(n^{\max\{3-\beta,2-\alpha,1-1/d\}+\eta})$ holds for all $\eta > 0$ we may conclude that it also holds with $\eta$ replaced by $o(1)$.
\end{proof}
%\begin{proof}
%By Lemma~\ref{lem:grid}, for $\mu = 2$, whp there are at most $n^{\max\{3-\beta,2-\alpha,1-1/d\}+o(1)}$ edges intersecting a grid of side length $1/2$, and two hyperplanes of this grid suffice to split $\Space$ into two halves. Whp there are $\Omega(n)$ vertices in each grid cell. Inspecting carefully the proofs in Section~\ref{sec:diameter}, we find that they still apply when we replace the torus $\Space$ by the space $[0,1]^{d-1} \times [0,1/2]$, without wrap-around. In particular, whp each grid cell contains a giant component of linear size. %Thus if we remove all edges intersecting the grid, then we split the giant component into $2^d$ components of linear size.
%\end{proof}

\paragraph{Compression algorithm:} With Lemma~\ref{lem:grid} at hand, we are ready to give a compression algorithm that stores the graph with $O(n)$ bits, i.e., with $O(1)$ bits per edge, proving Theorem~\ref{thm:entropy}. 
We remark that our result does not directly follow from the general compression scheme on graphs with small separators in~\cite{blandford2003compact}, since our graphs only have small separators in expectation, in particular, small subgraphs of size $O(\sqrt{\log n})$ can form expanders and thus not have small separators. However, our algorithm loosely follows their algorithm as well as the practical compression scheme of~\cite{BoldiV03}, see also~\cite{chierichetti2009compressing}.

We first enumerate the vertices as follows. 
Recall the definition of cells from Section~\ref{subsec:notation}, and consider all cells of volume $2^{-\ell_0 d}$, where $\ell_0 := \lfloor \log n /d\rfloor$. Note that the boundaries of these cells induce a grid as in Lemma~\ref{lem:grid}.
Since each such cell has volume $\Theta(1/n)$, the expected number of vertices in each cell is constant. We fix a geometric ordering of these cells as in Lemma~\ref{lem:geoorder}, and we enumerate the vertices in the order of the cells, breaking ties (between vertices in the same cell) arbitrarily. For the rest of the section we will assume that the vertices are enumerated in this way, i.e., we identify $V = [n]$, where $i \in [n]$ refers to the vertex with index $i$.

Having enumerated the vertices, for each vertex $i \in [n]$ we store a block of $1+\deg(i)$ sub-blocks. The first sub-block consists of a single dummy bit (to avoid empty sequences arising from isolated vertices). In the other $\deg(i)$ sub-blocks we store the differences $i-j$ using $\log_2 |i-j|+O(1)$ bits, where $j$ runs through all neighbors of $i$. We assume that the information for all vertices is stored in a big successive block $B$ in the memory. Moreover, we create two more blocks $B_V$ and $B_E$ of the same length. Both $B_V$ and $B_E$ have a one-bit whenever the corresponding bit in $B$ is the first bit of the block of a vertex, and $B_E$ has also a one-bit whenever the corresponding bit in $B$ is the first bit of an edge (i.e., the first bit encoding a difference $i-j$). All other bits in $B_V$ and $B_E$ are zero.

It is clear that with the data above the graph is determined. To handle queries efficiently, we replace $B_V$ and $B_E$ each with a rank/select data structure. This data structure allows to handle in constant time queries of the form ``Rank($b$)'', which returns the number of one-bits up to position~$b$, and ``Select($i$)'', which returns the position of the $i$-th one-bit~\cite{jacobson1989space,clark1996efficient,puatracscu2008succincter}. Given $i,s \in \N$, we can find the index of the $s$-th neighbor of $i$ in constant time by Algorithm~\ref{alg:findneighbor}, and the degree of $i$ by Algorithm~\ref{alg:finddegree}. In particular, it is also possible for Algorithm~\ref{alg:findneighbor} to first check whether $s \leq \deg(i)$. 

\begin{algorithm}
\caption{Finding the $s$-th neighbor of vertex $i$}\label{alg:findneighbor}
\begin{algorithmic}[1]
  \State $b := \Select(i,B_V)$ 
   \Comment{starting position of vertex $i$}
  \State $k := \Rank(b,B_E)$ \Comment{number of edges and vertices before $b$}
  \State $b_1:= \Select(k+s,B_E)$ \Comment{starting position of $s$-th edge of vertex $i$}
  \State $b_2 := \Select(k+s+1,B_E)$ \Comment{bit after ending position of $s$-th edge of vertex $i$}
  \State return $B[b_1:b_2-1]$ \Comment{block that stores $s$-th edge of vertex $i$}
 \end{algorithmic}
\end{algorithm}

\begin{algorithm}
\caption{Finding the degree of vertex $i$}\label{alg:finddegree}
\begin{algorithmic}[1]
  \State $b := \Select(i,B_V)$ \Comment{starting position of vertex $i$}
  \State $b' := \Select(i+1,B_V)$ \Comment{starting position of vertex $i+1$}
  \State $\Delta := \Rank(b',B_E)-\Rank(b,B_E)$ \Comment{block in $B_E$ contains $\deg(i)+1$ one-bits}
  \State return $\Delta - 1$
 \end{algorithmic}
\end{algorithm}

We need to show that the data structure needs $O(n)$ bits in expectation. There are $n$ dummy bits, so we must show that we require $O(n)$ bits to store all differences $i-j$, where $ij$ runs through all edges of the graph. We need $2\log_2|i-j|+O(1)$ bits for each edge, as we store the edge $ij$ both in the block of $i$ and in the block of $j$. The $O(1)$ terms sum up to $O(|E|)$, which is $O(n)$ in expectation. Thus, it remains to prove the following. 
\begin{lemma}\label{lem:compression}
Let the vertices in $V$ be enumerated by the  geometric ordering. Then,
\begin{equation}\label{eq:compression}
\Ex\left[\sum_{ij \in E} \log(|i-j|)\right] = O(n).
\end{equation}
\end{lemma}
\begin{proof}
We abbreviate the expectation in (\ref{eq:compression}) by $R$.
Note that the geometric ordering puts all the vertices that are in the same cell of a $2^{-\ell}$-grid in a consecutive block, for all $1 \leq \ell \leq \ell_0$. Therefore, if $e = ij$ does not intersect the $2^{-\ell}$-grid then $|i-j| \leq \#\{\text{vertices in the cell of $e$}\}$. %, which is $2+(n-2)2^{-d\ell} \leq 3n2^{-d\ell}$ in expectation. 
For $1\leq \ell \leq \ell_0$, let $\mathcal{E}_\ell$ be the set of edges intersecting the $2^{-\ell}$-grid. For convenience, let $\mathcal{E}_{0} := \emptyset$, and let $\mathcal{E}_{\ell_0+1} := E$ be the set of all edges. Then, using concavity of $\log$ in the third step,
\begin{align*}
R & \leq \Ex\left[\sum_{\ell=0}^{\ell_0}\sum_{e=ij\in \mathcal{E}_{\ell+1}\setminus \mathcal{E}_{\ell}}\log(\#\{\text{vertices in the cell of $e$}\}) \right]\\
& = \sum_{\ell=0}^{\ell_0}\sum_{u<v}\Pr\left[uv\in \mathcal{E}_{\ell+1}\setminus \mathcal{E}_{\ell}\right]\Ex\left[\log(\#\{\text{vertices in the cell of $u$}\}) \mid uv \in \mathcal{E}_{\ell+1}\setminus \mathcal{E}_{\ell}\right] \\
& \leq \sum_{\ell=0}^{\ell_0}\sum_{u<v}\Pr\left[uv\in \mathcal{E}_{\ell+1}\right]\log\bigg(\underbrace{\Ex\left[\#\{\text{vertices in the cell of $u$}\} \mid uv \in \mathcal{E}_{\ell+1}\setminus \mathcal{E}_{\ell}\right]}_{=: T_\ell}\bigg)
\end{align*}

The term $T_\ell$ is at most $T_\ell \leq 2+(n-2)2^{-\ell d} \leq 3n2^{-\ell d}$ for $\ell \leq \ell_0$ (where we count 2 for $u$ and $v$ and use independence of the other vertex positions). Thus it remains to show that $\Ex\big[ \sum_{\ell=0}^{\ell_0}|\mathcal{E}_{\ell+1}|\log(3n2^{-\ell d}) \big] = O(n)$. Let $\eta>0$ be sufficiently small. From Lemma~\ref{lem:grid} we know that $\Ex[|\mathcal{E}_{\ell}|] \leq E_\ell$, where $E_\ell= n \cdot (2^{d\ell}/n)^{\beta-2-\eta} + (n^{2- \alpha}2^{d\ell(\alpha -1)}+ n^{1-1/d}2^{\ell})(1+\log(n2^{-d\ell}))$. Since $E_\ell$ increases exponentially in~$\ell$, we obtain
\[
\Ex\left[\sum_{\ell=0}^{\ell_0}|\mathcal{E}_{\ell+1}|\log(3n2^{-d\ell})\right] \leq O\left(\sum_{\ell=1}^{\ell_0+1}E_\ell \log(3n2^{-d\ell+d})\right) = O(E_{\ell_0+1}\log(3n2^{-d\ell_0})) = O(n),
\]
where the last equality follows since $1/n \leq 2^{-d\ell_0} \leq O(1/n)$ by our choice of $\ell_0$. This proves the lemma, and hence shows that we need $O(n)$ bits in expectation to store the graph.
\end{proof}

This concludes the proof of Theorem~\ref{thm:entropy}, and it only remains to verify Lemma~\ref{lem:grid}. In the proof of this lemma, we will use the following technical statement, which is a consequence of Fubini's theorem and allows us to replace certain sums by integrals.

\begin{lemma}\label{lem:weightsums}
Let $f:\R\to\R$ be a continuously differentiable function. Then for any weights $0 \leq w_0 \leq w_1$,
\[
\sum_{v \in V, w_0 \leq \w{v} \leq w_1} f(\w{v}) \;=\; f(w_0)\cdot |V_{\geq w_0}| \;-\; f(w_1) \cdot |V_{> w_1}| \;+\; \int_{w_0}^{w_1} f'(w) \cdot |V_{\geq w}|  dw.
\]
In particular, if $f(0)=0$, then
\[\sum_{v \in V} f(\w{v}) = \int_{0}^{w_1} |V_{\geq w}| f'(w) dw=\int_{0}^{\infty} |V_{\geq w}| f'(w) dw.\]
\end{lemma}
\begin{proof}
Let $\nu$ be the sum of all Dirac measures given by the vertex weights between $w_0$ and $w_1$, i.e., for every set $A \subseteq \R$ we put $\nu(A):=|\{v \in V: \w{v} \in A, w_0 \le \w{v} \le w_1\}|$. Then
\begin{align*}
\sum_{v \in V, w_0 \leq \w{v} \leq w_1} f(\w{v}) &\;=\; \int_0^{\wmax} f(w) d\nu(w) \;=\; \int_0^{\wmax} \int_0^w  f'(x) dx d\nu(w) \;+\; \int_0^{\wmax} f(0) d\nu(w)\\ &\;=\; \int_0^{\wmax} \int_0^{\infty}  f'(x) \cdot \mathds{1}_{\{x \le w\}} dx d\nu(w) \;+\; f(0) \cdot \vert V_{\geq w_0} \setminus V_{>w_1}\vert.
\end{align*}
Notice that $[0,\wmax]$ is a compact set and $f'(x)$ is continuous by assumption. Hence the function $|f'(x)\cdot \mathds{1}_{\{x \le w\}}|$ is globally bounded on $[0,\wmax]$ and always zero for $x > \wmax$. Thus, $f'(x)\cdot \mathds{1}_{\{x \le w\}}$ is integrable and we can apply Fubini's theorem (see, e.g., \cite{evansmeasure}), which yields
\begin{align*}
\sum_{v \in V, w_0 \leq \w{v} \leq w_1} f(\w{v}) & \;=\; \int_0^{\infty} f'(x) \int_0^{\wmax}\mathds{1}_{\{w \ge x\}}d\nu(w) dx \;+\; f(0) \cdot \vert V_{\geq w_0} \setminus V_{>w_1}\vert \\
& \;=\; \int_0^{\infty} f'(x) \cdot \vert V_{\geq \max\{x,w_0\}} \setminus V_{>w_1}\vert dx \;+\; f(0) \cdot \vert V_{\geq w_0} \setminus V_{>w_1}\vert.
\end{align*}
The first summand (i.e.\ the integral) can be rewritten as
\[\int_0^{w_0} f'(x) \cdot \vert V_{\geq w_0} \setminus V_{>w_1}\vert dx \;\;+\; \int_{w_0}^{w_1} f'(x) \cdot \vert V_{\geq x} \setminus V_{>w_1}\vert dx ,\]
and then we obtain the first statement by using $\vert V_{\geq w_0} \setminus V_{>w_1}\vert=\vert V_{\geq w_0} \vert-\vert V_{>w_1}\vert$, calculating the integrals, and combining the resulting terms. The second statement follows directly by choosing $w_0=0$ and $w_1 > \wmax$.
%
%\\
%& \;=\; \int_0^{w_0} f'(x) \cdot \vert V_{\geq w_0} \setminus V_{>w_1}\vert dx \\
%& \qquad\;+\; \int_{w_0}^{w_1} f'(x) \cdot \vert V_{\geq x} \setminus V_{>w_1}\vert dx \;+\; f(0) \cdot \vert V_{\geq w_0} \setminus V_{>w_1}\vert\\
%& \;=\; f(w_0)\cdot |V_{\geq w_0}| \;-\; f(w_1) \cdot |V_{> w_1}| \;+\; \int_{w_0}^{w_1} f'(w) \cdot |V_{\geq w}|  dw.
%\end{align*}
\end{proof}

\begin{proof}[Proof of Lemma~\ref{lem:grid}]
%For ease of exposition, we will call the direction perpendicular to $h$ \emph{horizontal}. 
%Moreover let $r_{\max} = \Theta(1)$ be the diameter of the torus $\Space$. 
We first observe that we can assume $\mu \ge 2$ as this implies the statement for smaller $\mu$ immediately. Thus let $2\le \mu \le n^{1/d}$, and consider an axis-parallel, regular grid with side length $1/\mu$ on $\Space$.
For $u,v \in V$, let $\rho_{u v}$ be the probability that the edge $uv$ exists and cuts the grid. Let $r_{\max} := 1/2$ be the diameter of $\Space$. We write
\begin{align} \label{eq:rhouv}
 \rho_{u v} = \int_{0}^{r_{\max}} \Pr[\|\x u - \x v\| = r] \cdot p_{uv}(r) \cdot \Pr[\x u, \x v \text{ in different cells of $\mu$-grid}]\, dr. 
\end{align}
Observe that $u$ and $v$ have distance $r$ with probability density $\Pr[\|\x u - \x v\| = r] = O(r^{d-1})$. Furthermore, setting $\gamma_{uv} := \min\{(\w{u}\w{v}/\W)^{1/d},r_{\max}\}$ we have 
$$p_{uv}(r) = \begin{cases} \Theta(1), & \text{if } r \ge \gamma_{uv}, \\ \Theta( (\gamma_{uv} / r)^{\alpha d}), & \text{otherwise.} \end{cases}$$ 
Additionally, in the case $\alpha = \infty$, by increasing $\gamma_{uv}$ by at most a constant factor we may assume $p_{uv}(r) = 0$ for all $r \geq \gamma_{uv}$. %The constant $1/2$ in the definition of $\gamma_{uv}$ is a rather arbitrary choice to avoid that $|\log \gamma_{uv}|$ vanishes. 
For the last term in (\ref{eq:rhouv}), for a fixed axis of $\Space$ consider the hyperplanes $\{h_i\}_{1 \leq i \leq \mu}$ of the grid perpendicular to that axis. If the edge $e= uv$ has length $\|\x u - \x v\| = r$, then after a random shift along the axis, the edge $e$ intersects one of the $h_i$ with probability at most $\min\{\mu r,1\}$. By symmetry of the underlying space, a random shift does not change the probability to intersect one of the~$h_i$, so any edge of length $r$ has probability at most $\min\{\mu r,1\}$ to intersect one of the~$h_i$. By the union bound over all (constantly many) axes, the probability for $u,v$ to lie in different cells of the grid is $O(\min\{\mu r,1\})$. 

Now we distinguish several cases. For $\gamma_{uv}> 1/\mu$ and $\alpha < \infty$, we may estimate
\begin{align}\label{eq:hitgrid1}
\rho_{uv} %& \leq O\left(\int_{0}^{r_{\max}} r^{d-1} \cdot \min\{r,1\} \cdot p_{uv}(r) dr\right) \nonumber
& \leq  O\bigg(\int_{0}^{1/\mu} r^{d-1}\cdot \mu r dr + \int_{1/\mu}^{\gamma_{uv}} r^{d-1} dr+ \underbrace{\int_{\gamma_{uv}}^{r_{\max}} r^{d-1-d\alpha}\gamma_{uv}^{d\alpha} dr}_{= O(\gamma_{uv}^d), \text{ since } d-d\alpha <0} \bigg) \leq O(\mu^{-d} + \gamma_{uv}^d) \leq O(\gamma_{uv}^d).
\end{align}
For $\gamma_{uv} > 1/\mu$ and $\alpha = \infty$, equation~\eqref{eq:hitgrid1} remains true, except that the third integral is replaced by $0$ by our choice of $\gamma_{uv}$. So in this case we still get $\rho_{uv} \leq O(\gamma_{uv}^d)$.

The case $\gamma_{uv} \leq 1/\mu$ is a bit more complicated. Again we consider first $\alpha <\infty$. Then we may bound
\begin{align}\label{eq:hitgrid2}
\rho_{uv} & \leq O\bigg(\underbrace{\int_{0}^{\gamma_{uv}} r^{d-1}\cdot \mu r dr}_{=:I_1} + \underbrace{\int_{\gamma_{uv}}^{1/\mu} r^{d-1}\cdot \mu r\cdot r^{-d\alpha}\gamma_{uv}^{d\alpha} dr}_{=:I_2}+ \underbrace{\int_{1/\mu}^{r_{\max}} r^{d-1-d\alpha}\gamma_{uv}^{d\alpha} dr}_{=:I_3} \bigg).
\end{align}
Similarly as before, $I_1 \leq O(\gamma_{uv}^{d+1}\mu)$ and $I_3 \leq O(\mu^{d\alpha-d}\gamma_{uv}^{d\alpha})$. Note that both terms are bounded from above by $O((\gamma_{uv} \mu)^{d\tilde \alpha}\mu^{-d})$, where $\tilde \alpha := \min\{\alpha, 1+1/d\}$, since $\gamma_{uv} \mu \leq 1$. For $I_2$, the inverse derivative of $r^{d-d\alpha}$ is either $\Theta(r^{1+d-d\alpha})$, or $\log r$, or $-\Theta(r^{1+d-d\alpha})$, depending on whether $1+d-d\alpha$ is positive, zero, or negative, respectively. Therefore, we obtain
\[
I_2 \leq 
\left\{
\begin{aligned}
&O(\gamma_{uv}^{d\alpha}\mu^{d\alpha-d}) &&=  O((\gamma_{uv}\mu)^{d\alpha}\mu^{-d}), && \text{if }d-d\alpha>-1,\\
&O(\gamma_{uv}^{d\alpha}\mu (\log(1/\mu)-\log(\gamma_{uv}))) && = O((\gamma_{uv}\mu)^{d+1}\mu^{-d}|\log(\gamma_{uv}\mu)|), && \text{if }d-d\alpha=-1,\\
&O(\gamma_{uv}^{d+1}\mu) && = O((\gamma_{uv}\mu)^{d+1}\mu^{-d}), && \text{if }d-d\alpha<-1.
\end{aligned}\right.
\]

In particular, we can upper-bound all terms (including $I_1$ and $I_3$) in a unified way by $O((\gamma_{uv} \mu)^{d\tilde \alpha}\mu^{-d})(1+|\log(\gamma_{uv} \mu)|)$. Moreover, since $\gamma_{uv} \geq (\wmin^2/\W)^{1/d} = \Omega(n^{-1/d})$, the second factor is bounded by $O(1+\log (n^{1/d}/\mu)) = O(1+\log (n/\mu^d))$. Also, in the case $\alpha = \infty$ the same calculation applies, except that $I_2$ and $I_3$ are replaced by $0$. Note that naturally $\tilde \alpha = 1+1/d$ for $\alpha = \infty$. So altogether we have shown that
\begin{equation*}
\rho_{uv} \leq 
\begin{cases}
O(\gamma_{uv}^d), & \text{if }\gamma_{uv} \geq 1/\mu, \\
O((\gamma_{uv} \mu)^{d\tilde \alpha}\mu^{-d}(1+\log(n^{d}/ \mu))), & \text{if } \gamma_{uv} \leq 1/\mu.
\end{cases}
\end{equation*}
Therefore, the expected number of edges intersecting the grid is in $O(S_1 + S_2)$, where
\[
S_1 := \sum_{u,v\in V,\,  \gamma_{uv} > 1/\mu} \gamma_{uv}^{d} \quad \text{ and } \quad S_2 := \sum_{u,v\in V,\,  \gamma_{uv} \leq 1/\mu} (\gamma_{uv}\mu)^{d\tilde \alpha}\mu^{-d}(1+\log(n^{d}/\mu)).
\]

Let $0<\eta' <\eta <\beta-2$ be (sufficiently small) constants. Then we may use the power-law assumption~(PL2), Lemma~\ref{lem:totalweight}, and Lemma~\ref{lem:weightsums} to bound $S_1$:
\begin{align*}
S_1 & \leq \sum_{u,v\in V,\,  \w{u}\w{v} > \W/\mu^{d}} \frac{\w{u}\w{v}}{\W} = \sum_{u\in V} \frac{\w{u}}{\W}\cdot \W_{\geq \W/(\mu^d\w{u})} \stackrel{\ref{lem:totalweight}}{\leq} O\left(\sum_{u\in V} \frac{\w{u}}{\W} \cdot n\left(\frac{\W}{\mu^d\w{u}}\right)^{2-\beta+\eta}\right) \\
& \leq O\left(\left(\frac{\mu^d}{n}\right)^{\beta-2-\eta}\sum_{u\in V} \w{u}^{\beta-1-\eta}\right) \stackrel{\ref{lem:weightsums}}{\le} O\left(\left(\frac{\mu^d}{n}\right)^{\beta-2-\eta}\int_{\wmin}^{\infty} nw^{1-\beta+\eta'}w^{\beta-2-\eta}dw\right) \\
& = O\left(n\cdot \left(\mu^d/n\right)^{\beta-2-\eta}\right).
\end{align*}
To tackle $S_2$, we again use Lemma~\ref{lem:weightsums}, let $\lambda_u := \W/(\w{u}\mu^d)$ and obtain
\begin{align}\label{eq:hitgridexpectation}
S_2':=\sum_{\substack{u,v\in V \\ \gamma_{uv} \leq 1/\mu}} (\gamma_{uv})^{d\tilde \alpha}& = \sum_{u\in V} \sum_{\substack{v \in V_{\leq \lambda_u}}} \left(\frac{\w{u}\w{v}}{\W}\right)^{\tilde \alpha} 
\stackrel{\ref{lem:weightsums}}{\leq} O\left( \sum_{u\in V} \left(\frac{\w{u}}{n}\right)^{\tilde \alpha}\int_{\wmin}^{\lambda_u} nw^{1-\beta+\eta} w^{\tilde \alpha-1} dw\right) 
%\nonumber\\
%& \leq  O\left(n^{1-\tilde \alpha} \sum_{u\in V} \w{u}^{\tilde \alpha}\int_{\wmin}^{\W/\w{u}} w^{\tilde \alpha-\beta+\eta}dw\right) .
\end{align}

Now we distinguish two cases, because the integral behaves differently for exponents larger or smaller than $-1$. If $\tilde \alpha \geq \beta -1$, then for $0< \eta' < \eta$ equation~\eqref{eq:hitgridexpectation} evaluates to
\begin{align*}
S_2' & \leq O\left( \sum_{u\in V} \left(\frac{\w{u}}{n}\right)^{\tilde \alpha}n\lambda_u^{1+\tilde \alpha-\beta+\eta} \right) 
= O\left(\frac{n^{2-\beta+\eta}}{\mu^{d(1+\tilde \alpha-\beta+\eta)}} \sum_{u\in V} \w{u}^{\beta -1-\eta}\right) \\
& \stackrel{\ref{lem:weightsums}}{\leq} O\left(\frac{n^{2-\beta+\eta}}{\mu^{d(1+\tilde \alpha-\beta+\eta)}} \int_{\wmin}^{\infty} nw^{1-\beta+\eta'}w^{\beta -2-\eta}dw\right) = O\left(n\frac{n^{2-\beta+\eta}}{\mu^{d(1+\tilde \alpha-\beta+\eta)}}\right).
\end{align*}
Therefore, $S_2 = \mu^{d\tilde \alpha-d}(1+\log(n^{d}/\mu)) S_2' \leq O(n\cdot (n/\mu^d)^{2-\beta+\eta})$, which is one of the terms in the lemma. On the other hand, if $\tilde \alpha < \beta -1$ then for $0<\eta < \beta-\tilde \alpha -1$ we obtain from~\eqref{eq:hitgridexpectation}
\begin{align*}
S_2' & \leq O\left(n^{1-\tilde \alpha} \sum_{u\in V} \w{u}^{\tilde \alpha} \right) \stackrel{\ref{lem:weightsums}}{\leq} O\left(n^{1-\tilde \alpha}\int_{\wmin}^{\infty} nw^{1-\beta+\eta} w^{\tilde \alpha-1} dw\right) \leq O\left(n^{2-\tilde \alpha}\right),
\end{align*}
and again $S_2 = \mu^{d\tilde \alpha-d}(1+\log(n^{d}/\mu)) S_2'$ corresponds to terms in the lemma after plugging in $\tilde \alpha$. This concludes the proof. 
\end{proof}

\section{Comparison with Hyperbolic Random Graphs} 
\label{sec:hyperbolic}

In this section we show that hyperbolic random graphs are a special case of GIRGs.
We start by defining hyperbolic random graphs.  There exist several different representations of hyperbolic geometry, all with advantages and disadvantages. For introducing this random graph model, it is most convenient to use the native representation. It can be described by a disk $H$ of radius $R$ around the origin~$0$, where the position of every point $x$ is given by its polar coordinates $(r_x,\theta_x)$. The model is isotropic around the origin. The hyperbolic distance between two points $x$ and $y$ is given by the non-negative solution $d=d(x,y)$ of the equation
\begin{equation} \label{eq:hypdist}
\cosh(d) = \cosh(r_x)\cosh(r_y)-\sinh(r_x)\sinh(r_y)\cos(\phi_x-\phi_y).
\end{equation}
In the following definition, we follow the notation introduced by Gugelmann et al.\ \cite{gugelmann2012random}.
\begin{definition} \label{def:hyperrand}
Let $\alpha_H>0, C_H\in \R,T_H>0,n\in \N$, and set $R=2\log n+C_H$. Then the random hyperbolic graph $G_{\alpha_H,C_H,T_H}(n)$ is a graph with vertex set $V=[n]$ and the following properties:

\begin{itemize}
\item Every vertex $v \in [n]$ independently draws random coordinates $(r_v,\phi_v)$, where the angle $\pi_v$ is chosen uniformly at random in $[0,2\pi)$ and the radius $r_v \in [0,R]$ is random with density $f(r) := \frac{\alpha_H\sinh(\alpha_H r)}{\cosh(\alpha_H R)-1}$.
\item Every potential edge $e=\{u,v\}$, $u,v \in [n]$, is independently present with probability
\[p_H(d(u,v)) = \left(1+e^{\frac{1}{2T_H}(d(u,v)-R)}\right)^{-1}.\]
\end{itemize}

In the limit $T_H \rightarrow 0$, we obtain the threshold hyperbolic random graph $G_{\alpha_H,C_H}(n)$, where every edge $e=\{u,v\}$ is present if and only if $d(u,v) \le R$.
\end{definition}

We will show that hyperbolic random graphs are almost surely contained in our general framework. To this end, we embed the disk of the native hyperbolic model into our model with dimension 1, hence we reduce the geometry of the hyperbolic disk to the geometry of a circle, but gain additional freedom as we can choose the weights of vertices. Notice that a single point on the hyperbolic disk has measure zero, so we can assume that no vertex has radius $r_v=0$. For the parameters, we put 
\[d:=1,\quad\beta := 2\alpha_H + 1,\quad\alpha := 1/T_H.\]
Furthermore, we define the mapping
\[\w{v} := e^{\frac{R-r_v}{2}} \quad\text{and}\quad \x{v} := \frac{\phi_v}{2\pi}.\]
Since this is a bijection between $H \setminus \{0\}$ and $[1,e^{R/2}) \times \mathds{T}^1$, there exists as well an inverse function $g(\w{u},\x{u})=(r_u,\phi_u)$. Finally for any two vertices $u \neq v$ on the torus, we set 
\[p_{uv} := p_H(d(g(\w{u},\x{u}),g(\w{v},\x{v}))).\] 
This finishes our embedding. The following lemma, which we prove near the end of this section, demonstrates that under this mapping almost surely the weights will follow a power law.

\begin{lemma} \label{lem:hyppowerlaw}
Let $\alpha_H > \frac12$. Then for all $\eta=\eta(n)=\omega(\frac{\log\log n}{\log n})$, with probability $1-n^{-\Omega(\eta)}$ the induced weight sequence $\w{}$ follows a power law with parameter $\beta=2\alpha_H+1$. 
\end{lemma}

Now we come to the main statement of this section. %, which is a more precise reformulation of Theorem~\ref{thm:hyperbolic}. 
In the following we assume that if we sample an instance of the hyperbolic random graph model, we first sample the radii, then the angles and at last the edges.

\begin{theorem} \label{thm:hyperbolicprecise}
Let $\alpha_H > \frac12, n \in \N$ and fix a set of radii $(r_1, \ldots, r_n) \in [0,R]^n$ inducing a power-law weight sequence $\w{}$ with parameter $\beta =2\alpha_H+1$. 
Then the random positions $\x u$ and the edge probabilities $p_{uv}(\x u, \x v)$ produced by our mapping satisfy the properties of the GIRG model, i.e., for fixed radii inducing power-law weights, hyperbolic random graphs are a special case of GIRGs.
%Then our mapping transforms every hyperbolic random graph on $n$ vertices with radii $(r_1, \ldots, r_n)$ into an instance of GIRGs.
\end{theorem}

Note that the precondition of Theorem~\ref{thm:hyperbolicprecise} on the weight sequence $\w{}$ holds for any $\eta=\eta(n)=\omega(\log\log n/\log n)$ with probability $1-n^{-\Omega(\eta)}$ by Lemma~\ref{lem:hyppowerlaw}. Therefore an instance of random hyperbolic graphs is almost surely included in our GIRG model with parameters as set above. In particular, any property that holds with probability $1-q$ for GIRGs also holds for hyperbolic random graphs with probability at least, say, $1-q-n^{-o(1)}$.

Before proving Lemma~\ref{lem:hyppowerlaw} and Theorem~\ref{thm:hyperbolicprecise}, we consider the following basic property of hyperbolic random graphs.
\begin{lemma} \label{lem:hypmaxradius}
Let $\alpha_H > \frac12$. Then with probability $1-n^{-\Omega(1)}$ every vertex has radius at least $r_0 := (1-\frac{1}{2\alpha_H})\log n$. Furthermore, for all $r=\omega(1), r \le R$ and $v \in V$, we have
\[\Pr[r_v \le r] = e^{-\alpha_H(R-r)}(1+o(1)).\]
\end{lemma}
\begin{proof}
Let $v \in V$. By the given density $f$ it follows immediately that
\begin{align*}
\Pr[r_v \le r]&=\int_0^r f(x)dx =\alpha_H \int_0^r \frac{\sinh(\alpha_H x)}{\cosh(\alpha_H R)-1} dx = \frac{\cosh(\alpha_H r)-1}{\cosh(\alpha_H R)-1}\\ &= e^{-\alpha_H(R-r)}(1+o(1)),
\end{align*}
where we used $\cosh(x) = \frac{e^x+e^{-x}}{2} = \frac{e^x}{2}(1+o(1))$ whenever $x=\omega(1)$. Now let $X_{r_0}$ be the random variable counting the vertices of radius at most $r_0$. We observe that the above expression for $\Pr[r_v \le r]$ implies \[\Ex[X_{r_0}] = n e^{-\alpha_H(R-r_0)}(1+o(1))=e^{-\alpha_H C_H}n^{1/2-\alpha_H}(1+o(1))=n^{-\Omega(1)}.\] By Markov's inequality, with probability $1-n^{-\Omega(1)}$ we have $X_{r_0}=0$.
\end{proof}

\begin{proof}[Proof of Lemma~\ref{lem:hyppowerlaw}]
For every vertex of the random hyperbolic graph, the radius is chosen independently and uniformly according to $f(r)$. Hence under our mapping, we sample weights independently. We will prove that we fulfil the prerequisites of Lemma~\ref{lem:sampleweights}. Let $0<\eps<1$. By Lemma~\ref{lem:hypmaxradius}, the probability that a vertex $v$ has radius at most $r \ge \eps \log n$ is $\Theta(e^{-\alpha_H(R-r)})$. Let $1 \le z \le o(n^{1-\eps/2})$. Then $R-2\log z \ge \eps \log n$, and
\begin{align*}
F(z) := \Pr[\w{v} \le z] &= 1-\Pr[r_v \le R-2\log z]=1-\Theta(e^{-2\alpha_H\log z})\\&=1-\Theta(z^{-2\alpha_H})=1-\Theta(z^{1-\beta}).
\end{align*}
Furthermore, for $z < 1$ we get
\[F(z):= \Pr[\w{v} \le z] = \Pr[r_v \ge R-2\log{z}] = 0.\]
Clearly, $F(.)$ is non-decreasing and therefore satisfies the preconditions of Lemma~\ref{lem:sampleweights} by taking $\wmin=1$. Then it follows from this lemma that the weight sequence $\w{}$ follows a power law with parameter $\beta$ with sufficiently high probability.
\end{proof}

\begin{proof}[Proof of Theorem~\ref{thm:hyperbolicprecise}]
Let us start by considering the sampling process of a random hyperbolic graph. First we sample the radii of the vertices, for which the precondition of the theorem assumes that they induce a power-law weight sequence. Next we sample the angles. This corresponds to coordinates chosen independently and uniformly at random on $\mathds{T}^1$. It remains to prove that $p_{uv}$ as defined above satisfies conditions (\ref{eq:puv}) and (\ref{eq:puv2}).

Let $u \neq v$ be two vertices of the random hyperbolic graph with coordinates $(r_u,\phi_u)$ and $(r_v,\phi_v)$ and consider their mappings $(\w{u},\x{u})$ and $(\w{v},\x{v})$. Since the hyperbolic model is isotropic around the origin, we can assume without loss of generality that $r_u \ge r_v$, $\phi_v=0$ and $\phi_u \le \pi$. 

Let us first consider the threshold model, corresponding to $\alpha = \infty$. We claim that there exist constants $M > m > 0$ such that whenever $\|\x{u}-\x{v}\|\ge M \frac{\w{u}\w{v}}{\W}$, then $p_{uv}=0$, and whenever $\|\x{u}-\x{v}\|\le m \frac{\w{u}\w{v}}{\W}$, then $p_{uv}=1$. This will imply (\ref{eq:puv2}), as we set $d=1$.
Recall that in the threshold model, two vertices $u$ and $v$ are connected if and only if $d(u,v)\le R$. When $r_u+r_v \le R$, this is the case for all angles $\phi_u$ and $\phi_v$. Otherwise, for $\phi=0$ and $\phi_u \le \pi$, the distance between $u$ and $v$ is increasing in $\phi_u$ and there exists a critical value $\phi_0$ such that $d((r_u,\phi_u),(r_v,0))\le R$ if and only if $\phi_u \le \phi_0$. The following lemma estimates $\phi_0$.
\begin{lemma}[Lemma~3.1 in \cite{gugelmann2012random}] \label{lem:criticalangle}
Let $0 \le r_u \le R$, $r_u+r_v \ge R$ and assume $\phi_v=0$. Then
\[\phi_0 \;=\; 2e^{\frac{R-r_u-r_v}{2}}\left(1+\Theta(e^{R-r_u-r_v})\right).\]
\end{lemma}

Suppose $\|\x{u}-\x{v}\|\ge M \frac{\w{u}\w{v}}{\W}$. Notice that by our transformation we have $\|\x{u}-\x{v}\|=\frac{\phi_u}{2\pi}$ and 
\begin{equation}  \label{eq:weightcorresp}
\frac{\w{u}\w{v}}{\W} = \Theta\left(\frac{\w{u}\w{v}}{n}\right)=\Theta\left(\frac{e^{R-(r_u+r_v)/2}}{n}\right) = \Theta(e^{(R-r_u-r_v)/2}),
\end{equation}
where we used $\W=\Theta(n)$ by Lemma~\ref{lem:totalweight}. Hence, we have $\frac{\phi_u}{2\pi M}=\Omega(e^{(R-r_u-r_v)/2})$. Since $\phi_u \le 1$, this implies $r_u+r_v > R$, if we choose the constant $M$ sufficiently large. Moreover, for sufficiently large $n$ we obtain $\phi_u > 2e^{\frac{R-r_u-r_v}{2}}\left(1+\Theta(e^{R-r_u-r_v})\right)$. Thus, by Lemma~\ref{lem:criticalangle} the two vertices $u$ and $v$ are not connected and indeed $p_{uv}=0$. 

On the other hand, assume $\|\x{u}-\x{v}\|\le m \frac{\w{u}\w{v}}{\W}$. Then either $r_u+r_v < R$ and thus $\{u,v\} \in E$ follows directly, or $r_u+r_v \ge R$ and $\phi_u < 2e^{\frac{R-r_u-r_v}{2}}\left(1+\Theta(e^{R-r_u-r_v})\right)$, if $m$ is sufficiently small. In the second case, Lemma~\ref{lem:criticalangle} implies $p_{uv}=1$.

We now turn to the case $\alpha < \infty$. By our assumptions on $\phi_u$ and $\phi_v$ and by the identity $\cosh(x \pm y)=\cosh(x)\cosh(y) \pm \sinh(x)\sinh(y)$, we can rewrite (\ref{eq:hypdist}) as
\begin{equation} \label{eq:newhypdist}
\cosh(d)=\cosh(r_u-r_v)+(1-\cos(\phi_u))\sinh(r_u)\sinh(r_v).
\end{equation}
Next we observe that $\cosh(x)=\Theta(e^{|x|})$ for all $x$ and $\sinh(x)=\Theta(e^x)$ for all $x=\omega(1)$. Observe that (PL2) and $\w v = e^{(R-r_v)/2}$ imply $r_v = \Theta(\log n)$ for all vertices $v$. Furthermore, we perform a Taylor approximation of $1-\cos(\phi_u)$ around $0$ and get
$1-\cos(\phi_u)=\frac{\phi_u^2}{2}-\frac{\phi_u^4}{24}+\ldots = \Theta(\phi_u^2)$,
as $\phi_u$ is at most a constant. Combining these observations with (\ref{eq:newhypdist}) and the assumption $r_u \ge r_v$, we deduce
\begin{equation} \label{eq:finalterm}
e^{d-R}=\Theta(\cosh(d)e^{-R})=\Theta\left(e^{r_u-r_v-R}+\phi_u^2 e^{r_u+r_v-R}\right).
\end{equation}
In the condition (\ref{eq:puv}) on $p_{uv}$ the minimum is obtained by the second term whenever $\|\x{u}-\x{v}\| \le \frac{\w{u}\w{v}}{\W}$. Mapping $u$ and $v$ to the hyperbolic disk, this implies $\phi_u =O(e^{(R-r_u-r_v)/2})$. We claim that whenever $\phi_u =O(e^{(R-r_u-r_v)/2})$, the two vertices $u$ and $v$ are connected with constant probability and therefore $p_{uv}=\Theta(1)$. Indeed, in this case by (\ref{eq:finalterm}) we have $e^{d-R}=O(1)$, and using Definition~\ref{def:hyperrand} we deduce
\[p_{uv}=p_H(d(u,v))=\left(1+(e^{d-R})^{(1/(2T_H))}\right)^{-1}=\Theta(1).\]
On the other hand, suppose $\|\x{u}-\x{v}\| \ge \frac{\w{u}\w{v}}{\W}$, which implies $\phi_u =\Omega(e^{(R-r_u-r_v)/2})$. In this case by (\ref{eq:finalterm}) we have
$e^{d-R}=\Theta\left(\phi_u^2 e^{r_u+r_v-R}\right)=\Omega(1)$.
However, if $e^{d-R}=\Omega(1)$, we can use Definition~\ref{def:hyperrand} and (\ref{eq:weightcorresp}) to obtain
\begin{align*}
p_{uv}&=\left(1+e^{\frac{1}{2T}(d-R)}\right)^{-1}=\Theta\left(e^{-\frac{1}{2T_H}(d-R)}\right)=\Theta\left(\left(\phi_u^{-1}e^{(R-r_u-r_v)/2}\right)^{1/T_H}\right)\\
&=\Theta\left(\left(\frac{1}{\|\x{u}-\x{v}\|}\cdot\frac{\w{u}\w{v}}{\W}\right)^{1/T_H}\right)=\Theta\left(\left(\frac{1}{\|\x{u}-\x{v}\|}\cdot\frac{\w{u}\w{v}}{\W}\right)^{\alpha}\right).
\end{align*}
This finishes the case $\alpha < \infty$ and thus the proof.
\end{proof}

\section{Conclusion}\label{sec:conclusion}
To cope with the technical shortcomings of hyperbolic random graphs, we introduced a new model of scale-free random graphs with underlying geometry -- geometric inhomogeneous random graphs -- and theoretically analyzed their fundamental structural and algorithmic properties. Scale-freeness and basic connectivity properties of our model follow from more general considerations~\cite{bringmann2015generalGIRG}. We established that (1) hyperbolic random graphs are a special case of GIRGs, (2) GIRGs have a constant clustering coefficient, and (3) GIRGs have small separators and are very well compressible. As our main result, (4) we presented an expected-linear-time sampling algorithm. This improves the best-known sampling algorithm for hyperbolic random graphs by a factor $O(\sqrt{n})$.

In this paper we laid the foundations for further experimental and theoretical studies on GIRGs. In particular, we hope that the model can be used for the analysis of processes such as epidemic spreading. We leave this to future work.

%The most important experimental finding for hyperbolic random graphs is that greedily constructed paths are very close to shortest paths~\cite{BogunaPK10}. Hence, we will study greedy routing on GIRGs in future work, for which we laid the foundations in the present paper.

\begin{paragraph}{Acknowledgements.}
We thank Hafsteinn Einarsson, Tobias Friedrich, and Anton Krohmer for helpful discussions.
\end{paragraph}

\begin{footnotesize}
\bibliographystyle{plain}
\bibliography{../girg}
\end{footnotesize}

\end{document}